\documentclass[10pt, doublecolumn]{IEEEtran}
\usepackage{epsfig,latexsym}
\usepackage{float}
\usepackage{indentfirst}
\usepackage{amsmath}
\usepackage{bm}
\usepackage{amssymb}
\usepackage{times}
\usepackage{enumitem}
\usepackage{cite}
\usepackage{lastpage}
\usepackage{color} 
\usepackage{amsthm}
\usepackage{array}
\usepackage{algorithm}
\usepackage{algorithmic}

\usepackage{multirow}
\usepackage{empheq}
\usepackage{subfigure}
\usepackage{tikz}
\usetikzlibrary{arrows.meta,calc,positioning}
\newtheorem{proposition}{Proposition}

\newtheorem{corollary}{Corollary}

\begin{document}
\title{Access Protocols for Segmented Waveguide-Enabled Pinching-Antenna Systems (SWANs)}
\author{Shan Shan, Chongjun Ouyang, Petar Popovski, and Yuanwei Liu 
\thanks{Shan Shan is with the School of Information and Communication Engineering, Beijing University of Posts and Telecommunications, Beijing 100876, China (email: shan.shan@bupt.edu.cn). Chongjun Ouyang is with the School of Electronic Engineering and Computer Science, Queen Mary University of London, London E1 4NS, U.K. (email: c.ouyang@qmul.ac.uk). Petar Popovski is with the Connectivity Section, Department of Electronic Systems, Aalborg University, 9220 Aalborg, Denmark (e-mail: petarp@es.aau.dk). Yuanwei Liu is with the Department of Electrical and Computer Engineering, The University of Hong Kong, Hong Kong (email: yuanwei@hku.hk).}}
\IEEEaftertitletext{\vspace{-3em}}
\maketitle
\begin{abstract}
This paper {proposes an access protocol framework for} segmented waveguide-enabled pinching-antenna systems (SWANs), which exploits SWAN-induced reconfigurable channel diversity as a protocol-level resource for uplink random access. 
{The framework consists of two stages, a channel-oracle stage and an access stage, designed under three SWAN operating modes: \emph{(i)} one-segment selection (OS), \emph{(ii)} segment aggregation (SA), and \emph{(iii)} segment multiplexing (SM).} 
Specifically, in the channel-oracle stage, the OS mode is adopted to acquire sparse pilot observations and infer the channel responses across the SWAN configuration space. In this way, high-dimensional uplink channel acquisition is recast as a low-dimensional geometric localization problem, thereby reducing pilot overhead while preserving channel reconstruction accuracy.
{For} the access stage, we construct two oracle-guided access codebooks under the SA and SM modes, respectively, which {address the} tradeoff between hardware complexity and multiuser access resolution.
In particular, the SA-based scheme supports single radio frequency (RF) chain access through randomized segment-group activation, whereas the SM-based $R$-access scheme exploits multiple RF chains to construct deterministic access slots and enhance collision resolution.
Finally, our numerical results demonstrate that \emph{(i)} the proposed two-stage framework improves access performance under the same training overhead, \emph{(ii)} anchor densification is more effective than aggressive segment aggregation for SA, and \emph{(iii)} SM-based $R$-access achieves deterministic coverage and higher throughput in moderate- and high-load regimes, whereas SA-based access remains attractive for low-complexity implementations.
\end{abstract}

\begin{IEEEkeywords}
Channel oracle, random access protocol, segmented waveguide-enabled pinching-antenna systems.
\end{IEEEkeywords}

\section{Introduction}
The evolution toward sixth-generation (6G) wireless networks is expected to support massive connectivity with high service reliability for emerging applications such as massive machine-type communications (mMTC) and industrial Internet of Things (IIoT)~\cite{8869705,10109667,10745245,9795904, chen2025channel}. This requirement makes random access indispensable for scalable uplink connectivity. It enables sporadically active and uncoordinated devices to initiate transmissions before explicit scheduling by the access point (AP) over shared wireless resources, such as time slots, frequency bandwidth, and preamble sequences~\cite{popovski2020wireless,7565189,4385793}. However, in conventional random access protocols, the wireless propagation environment is usually regarded as determined and uncontrollable. 
As a result, the access opportunity of each active user is strongly influenced by its static geometric relationship with the AP~\cite{chen2023energy}. Specifically, users with advantageous link conditions tend to experience persistently stronger access channels, whereas edge or geometrically blockaged users are subject to reduced access success probability. Such persistent link imbalance fundamentally limits the reliability and fairness of random access under massive connectivity.
% Such link imbalance becomes a major bottleneck for effective random access under massive connectivity.

This limitation raises the question of whether the wireless propagation environment itself can be actively controlled and exploited for access protocol design beyond conventional wireless resources. Motivated by this, we consider pinching-antenna systems (PASS), a representative flexible-antenna architecture, to enable reconfigurable propagation behavior in wireless infrastructures~\cite{10945421, 11364174}. Specifically, PASS employs low-loss dielectric waveguides to deliver electromagnetic signals, along which multiple pinching antennas (PAs) can be activated as reconfigurable radiation or collection elements. By activating different PA locations along the deployed waveguide structure, PASS can effectively adjust the wireless emission or collection point, thereby establishing stable line-of-sight (LoS) links for different users even in obstructed environments. This is particularly beneficial for random access, since PA reconfiguration alleviates the persistent link imbalance of fixed-antenna systems and recasts PA configuration-dependent propagation variations to serve as an additional access resource for sporadically active users.

\subsection{Related Works}
% Recent studies have validated the effectiveness of PASS in improving physical-layer performance. For example, the pioneering theoretical study in~\cite{10945421} analyzed the fundamental transmission characteristics of PASS. 
Recent studies have validated the effectiveness of PASS in improving physical-layer performance. The pioneering theoretical study in~\cite{10945421} analyzed the fundamental transmission characteristics of PASS, while subsequent works studied downlink rate maximization~\cite{10896748} and uplink minimum-rate maximization~\cite{10909665}. More general modeling and beamforming optimization frameworks were further developed in~\cite{11202577,11263923,11223640,11300296}, which demonstrated the benefits of jointly exploiting reconfigurable PA locations and waveguide-domain propagation for physical-layer transmission design.

To further exploit its programmable propagation capability, several PASS architectures have been investigated. For example, segmented waveguide-enabled PASS (SWAN) divides a long waveguide into multiple short segments with dedicated feed points, which helps avoid inter-antenna re-radiation (IAR) and provides flexible segment-level control~\cite{11348983,jiang2026swan}. Center-fed PASS (C-PASS) introduces more flexible signal routing through waveguide power splitting and direction switching~\cite{xu2026cpass,gan2026cpass}, while multi-mode multi-feed PASS (M-PASS) exploits multi-mode signal propagation to facilitate multiuser multiplexing over waveguides~\cite{xiaoxia2026mpass}. In addition, tri-hybrid PASS architectures have been developed by employing multiple waveguides and jointly exploiting waveguide selection, PA activation, and beamforming design, thereby further expanding the configurable propagation space~\cite{cheng2025trihybrid,zhao2025fullyconnected}. These architectural studies show that PASS is a configurable propagation paradigm with different levels of spatial controllability, hardware complexity, and protocol compatibility.

Beyond performance analysis and beamforming optimization, PASS has also been investigated for physical-layer multiple access and multiuser interference management. Existing studies first incorporated conventional multiple access strategies into PASS-enabled multiuser transmission. For example, PASS-based non-orthogonal multiple access (NOMA) and orthogonal multiple access (OMA) designs were studied in~\cite{yanyu2025omanoma,ren2025passma}, where PA locations and resource allocation were jointly optimized for coordinated multiuser communication. More recently, PASS-specific multiple access schemes have been proposed to further exploit the waveguide-domain and propagation-reconfiguration capabilities of PASS. Specifically, the authors of~\cite{11435305} proposed waveguide-division multiple access (WDMA), where dedicated waveguides are assigned to different users and pinching beamforming is exploited for desired-signal enhancement and interference mitigation. In~\cite{zhiguo2025edma}, environment-division multiple access (EDMA) utilized PA reconfiguration to reshape line-of-sight (LoS) links and large-scale path losses for interference control. Furthermore, PASS-enabled multiple access for multicast transmission was further studied in~\cite{shan2026multigroupma,11442795}, where waveguide-domain transmission structures and PA placement were jointly designed to improve multiuser fairness. These studies demonstrate the potential of PASS for coordinated multiuser transmission by exploiting PA placement, waveguide resources, and propagation reconfiguration. Nevertheless, they mainly focus on scheduled multiple access, where the active users are known and PA configurations, power allocation, beamforming, or access resources can be centrally optimized based on available user and channel information.
\vspace{-5pt}
\subsection{Motivation and Contributions}
Less attention has been paid to how PASS should be integrated into higher-layer protocols, notably the access protocols for shared medium. Random access is relevant for uplink PASS communications from uncoordinated users with sporadic activity, whose channels are not known a priori. The key challenge is therefore how to exploit the configurable PASS propagation space to coordinate contention-based uplink access efficiently. This gives rise to multiple specific challenges: utilize  the channel variations induced by different PA configurations, optimize training overhead, user  access-slot selection, as well as balance between access reliability and implementation complexity.

To study this problem under a physically consistent and practically implementable architecture, this article focuses on SWANs. In this architecture, a long waveguide is divided into multiple short dielectric segments with dedicated feed points, and one PA can be activated on each segment during uplink reception. 
SWAN naturally supports three modes of hardware complexity and access controllability: one-segment selection (OS), segment aggregation (SA), and segment multiplexing (SM). These features make SWAN a particularly suitable architecture for investigating uplink random access over PASS.
Differently from the traditional problem of contention, the random access protocol design for SWAN cannot be separated from channel acquisition. This is because the access decisions are tightly coupled with the segmented architecture, the operating mode, and the user-dependent channel responses induced by different segment-pinch configurations. Sampling too many configurations would incur prohibitive overhead, whereas relying on overly coarse access rules would fail to exploit the structured channel variation enabled by SWAN. Instead, one should extract sufficient architecture-aware channel knowledge from sparse observations and then use it to guide access over a large configuration space with addressable complexity. 

% TODO Contributions
This article develops a two-stage oracle-assisted random access framework aligned with the three canonical SWAN operating modes. In the \emph{oracle stage}, the OS mode is employed to acquire sparse pilot observations and to learn a predictive channel oracle, which infers the channel responses over the entire SWAN configuration space without requiring exhaustive sampling or perfect a priori CSI. In the \emph{access stage}, {the reconstructed channel information is then used to support either single-radio frequency (RF) SA-based access, or SM-based $R$-access, a multi-RF option with stronger controllability.} In this way, oracle learning and access scheduling are coupled through the same SWAN geometry, while the two access architectures expose explicit tradeoffs among training overhead, access reliability, and hardware complexity.

The main contributions of this paper are summarized as follows. {\emph{(i)} We introduce a principled mechanism for synthesizing the complete channel responses from sparse pilot observations rather than from exhaustive per-configuration training. This is based on an  OS-based channel-oracle module for SWAN, in which the high-dimensional configuration-dependent channel acquisition problem is recast as a low-dimensional geometric parameter estimation problem. \emph{(ii)} We develop and analyze an SA-based random-access scheme for the single-RF architecture, including a group-based access codebook. This translates oracle knowledge into statistical access decisions under analog coherent aggregation. \emph{(iii)} We design an SM-based $R$-access scheme for the multi-RF architecture, in which oracle knowledge is converted into deterministic slot-wise access decisions. A geometry-aware lower bound is derived to guarantee that each user has at least one feasible access slot under the proposed access codebook.}
{Our numerical results validate the proposed oracle and access designs and to quantify the architectural tradeoff between SA-based access and $R$-access. The results show that geometry-aware sparse sample improves oracle quality, that moderate SA grouping is preferable, and that while $R$-access achieves higher throughput and deterministic coverage, SA remains attractive under a constrained RF chain budget.} 
\vspace{-5pt}
\subsection{Organization and Notation}
The rest of this paper is organized as follows. Section~\ref{sec_SWAN_protocols} introduces the SWAN architecture, the SWAN-based channel model, and the basic operating modes. Section~\ref{sec_SWAN_SS_oracle} presents the OS-based channel-oracle module. Section~\ref{sec_SWAN_access} develops the SA-based access module, while Section~\ref{sec_partialSM_access} presents the SM-based $R$-access module. Section~\ref{sec_numerical_results} provides numerical results. Finally, Section~\ref{sec_conclusion} concludes the paper.

\subsubsection*{Notations}
Scalars, vectors and matrices are represented by italic letters, bold lowercase letters and bold uppercase letters, respectively. For a vector ${\mathbf x}$, $[{\mathbf x}]_i$, ${\mathbf x}^{\mathsf T}$, ${\mathbf x}^{\mathsf H}$, and $\|{\mathbf x}\|_2$ denote its $i$th element, transpose, Hermitian transpose, and Euclidean norm, respectively. The cardinality of a set $\mathcal A$ is denoted by $|\mathcal A|$. The statistical expectation, probability, variance, covariance, and real-part operators are denoted by $\mathbb E\{\cdot\}$, $\Pr(\cdot)$, $\mathrm{Var}(\cdot)$, $\mathrm{Cov}(\cdot)$, and $\Re\{\cdot\}$, respectively. ${\mathcal C}{\mathcal N}(\boldsymbol\mu,\mathbf C)$ denotes a circularly symmetric complex Gaussian distribution with mean $\boldsymbol\mu$ and covariance matrix $\mathbf C$, while $\mathrm{Binomial}(N,p)$ and $\mathrm{Poisson}(\lambda)$ denote the binomial and Poisson distributions, respectively. The symbols ${\mathbf 0}$ and ${\mathbf I}$ denote the all-zero vector and identity matrix, respectively. The floor and ceiling operators are denoted by $\lfloor\cdot\rfloor$ and $\lceil\cdot\rceil$, respectively. Finally, ${\mathbf 1}\{\cdot\}$ denotes the indicator function.

\section{System Model}\label{sec_SWAN_protocols}
{This section introduces the SWAN-assisted uplink random-access model considered in this paper. We consider one access point (AP) equipped with a SWAN front-end and multiple single-antenna UEs with sporadic uplink traffic. At the beginning of a protocol period, the AP does not have a priori knowledge of the active UEs' locations or PA configuration-dependent channel coefficients. The AP controls the segment switches, feed connections, and PA configurations, whose control signaling is assumed not to interfere with the downlink pilot or uplink packet transmissions. }
% The time required to load a new SWAN configuration is captured by the switching overhead $T_{\mathrm{sw}}$ introduced later in the protocol timing model.

\subsection{SWAN Architecture}\label{subsec_SWAN_architecture}
\begin{figure}[!t]
\centering
\includegraphics[width=0.5\textwidth]{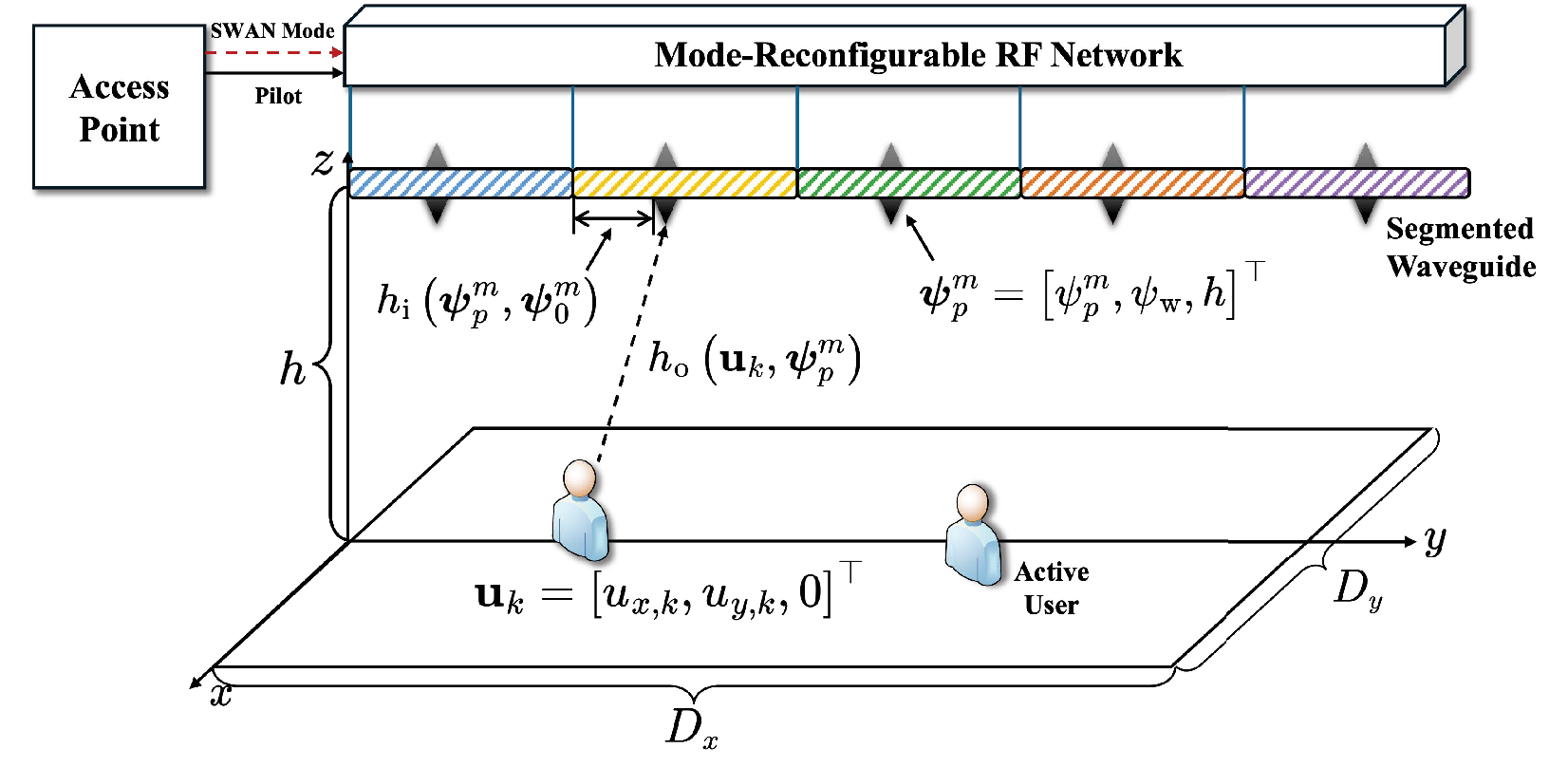}\vspace{-5pt}
\caption{Illustration of the proposed SWAN architecture.}
\label{Figure_PAS_System_Model}\vspace{-10pt}
\end{figure}

\begin{figure*}[!t]
\centering
\includegraphics[width=0.95\textwidth]{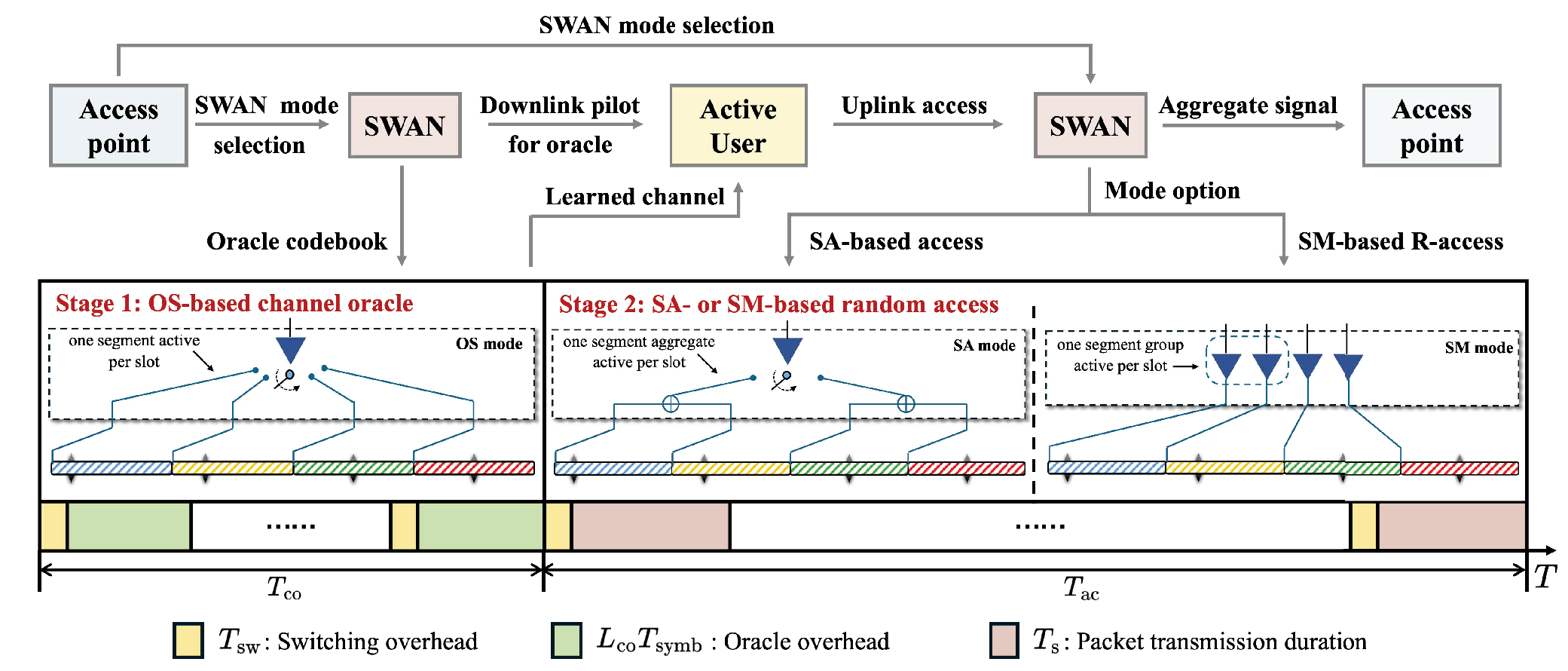}
\caption{Illustration of the proposed two-stage oracle-assisted random access protocol for SWAN. In Stage~1, the OS-based channel oracle uses downlink pilots to learn the configuration-dependent channels. In Stage~2, the learned channel information guides uplink random access under either SA-based access or SM-based $R$-access. The timeline shows the switching overhead $T_{\rm sw}$, oracle pilot overhead $L_{\rm co}T_{\rm symb}$, and packet transmission duration $T_{\rm s}$.}
\label{PASS_RA}
% \vspace{-10pt}
\end{figure*}

{% The region $\mathcal U$ is assumed to be known at the AP and serves as the geometric domain for channel oracle learning and access design. 
As illustrated in Fig.~\ref{Figure_PAS_System_Model}, the SWAN architecture consists of $M$ dielectric waveguide segments arranged along the $x$-axis, each of length $L$ such that $D_x=ML$. Unlike conventional PASS with a single continuous waveguide, the $M$ segments are not physically interconnected. Instead, each segment has an independent feed point and is connected to the AP RF front-end through low-loss wired links, whose loss is neglected.}
{Let $\mathcal M\triangleq\{1,\ldots,M\}$ denote the segment index set. The feed point of the $m$th segment, $m\in\mathcal M$, is located at $\boldsymbol{\psi}_0^m \triangleq [\psi_0^m,\psi_{\mathrm w},h]^{\mathsf T}$, where $\psi_0^1<\cdots<\psi_0^M$, $\psi_{\mathrm w}$ is the fixed $y$-coordinate of the waveguide, and $h$ is the deployment height. Within the $m$th segment, the AP can activate one radiating/receiving point from a discrete set of $P$ candidate PA locations. Let $\mathcal P_m\triangleq\{1,\ldots,P\}$ denote the corresponding index set. The location of the $p$th PA, $p\in\mathcal P_m$, is given by}
\begin{align}
{\bm\psi}_p^m \triangleq [\psi_p^m,\psi_{\mathrm w},h]^{\mathsf T},\qquad \psi_0^m\le \psi_p^m \le \psi_0^m+L .
\end{align}

{In the uplink of SWAN, a signal captured by one PA may propagate along the waveguide segment toward the feed point and be re-radiated into free space through another activated PA on the same segment, which is referred to as the inter-antenna re-radiation (IAR) effect~\cite{11348983}. Since IAR substantially complicates the uplink signal model, we assume that only a single PA is activated per segment. Accordingly, an elementary PA configuration is specified by}
\begin{align}\label{eq_config_pair}
c\triangleq (m,p),\qquad m\in\mathcal M,\ p\in\mathcal P_m,
\end{align}
{and the corresponding configuration space is defined as}
\begin{align}\label{eq_config_space}
\mathcal C \triangleq \{(m,p):m\in\mathcal M,\ p\in\mathcal P_m\},\qquad |\mathcal C|=MP .
\end{align}
{For multi-segment operating modes, a slot-level SWAN configuration is obtained by assigning one such elementary configuration to each activated segment according to the selected operating mode.}
\vspace{-10pt}
\subsection{SWAN-Based Channel Model}\label{subsec_SWAN_channel_model}
{As illustrated in Fig.~\ref{Figure_PAS_System_Model}, the UEs are distributed in a rectangular service region of size $D_x\times D_y$ on the $xy$-plane. The location of the $k$th UE is denoted by}
\begin{align}
{\mathbf u}_k=[u_{x,k},u_{y,k},0]^{\mathsf T}\in\mathcal U\triangleq [0,D_x]\times[0,D_y].
\end{align}
{We adopt a free-space line-of-sight (LoS) model for the radiated UE-PA link, which is appropriate for the high-frequency SWAN deployment and provides a geometry-explicit basis for oracle learning. The uplink LoS coefficient between the $k$th UE and the PA located at ${\bm\psi}_p^m$ can be written as follows:}
\begin{align}\label{eq_ho_def}
h_{\mathrm o}({\mathbf u}_k,{\bm\psi}_p^m)
\triangleq
\frac{\eta^{\frac{1}{2}}{\rm e}^{-{\rm j} k_0\|{\mathbf u}_k-{\bm\psi}_p^m\|}}{\|{\mathbf u}_k-{\bm\psi}_p^m\|},
\end{align}
where $\eta=\frac{c^2}{16\pi^2 f_{\mathrm c}^2}$ and $k_0=2\pi/\lambda$. The in-segment propagation coefficient from the PA located at ${\bm\psi}_p^m$ to the feed point ${\bm\psi}_0^m$ is given by
\begin{align}\label{eq_hi_def}
h_{\mathrm i}({\bm\psi}_p^m,{\bm\psi}_0^m)
\triangleq
10^{-\frac{\kappa}{20}|\psi_p^m-\psi_0^m|}
{\rm e}^{-{\rm j}\frac{2\pi}{\lambda_{\mathrm g}}|\psi_p^m-\psi_0^m|},
\end{align}
{where $\kappa$ is the attenuation factor (dB/m) and $\lambda_{\mathrm g}$ is the guided wavelength. Under the one-PA-per-segment activation priciple, the segment-wise effective uplink channel coefficient under PA configuration $(m,p)$ can be written as follows:}
% The geometry-dependent LoS model serves as the parametric basis of the subsequent oracle design. Thus, the oracle observations are assumed to remain compatible with \eqref{eq_ho_def}, while fully blocked segment-UE links are treated as unavailable or weak-signal observations. 
\begin{align}\label{eq_zetaUL_def}
\zeta_k^{\mathrm{ul}}(m,p) \triangleq h_{\mathrm i}({\bm\psi}_p^m,{\bm\psi}_0^m)\, h_{\mathrm o}({\mathbf u}_k,{\bm\psi}_p^m)\in\mathbb C.
\end{align}
{By channel reciprocity, the corresponding downlink coefficient used for oracle pilot transmission is given by}
\begin{align}\label{eq_zetaDL_def}
\zeta_k^{\mathrm{dl}}(m,p)=\zeta_k^{\mathrm{ul}}(m,p).
\end{align}

\subsection{SWAN Operating Modes}\label{subsec_SWAN_three_protocols}
% \begin{figure}[!t]
% \centering
% \subfigure[One-segment selection (OS) for oracle module.]{
%     \includegraphics[width=0.4\textwidth]{figure/System_model_OS.pdf}
%     \label{Figure_PAN_Protocol1}
% }\\[4pt]
% \subfigure[Segment aggregation (SA) for SA-based access module.]{
%     \includegraphics[width=0.4\textwidth]{figure/System_model_SA.pdf}
%     \label{Figure_PAN_Protocol2}
% }\\[4pt]
% \subfigure[Segment multiplexing (SM) for $R$-access module.]{
%     \includegraphics[width=0.4\textwidth]{figure/System_model_SM.pdf}
%     \label{Figure_PAN_Protocol3}
% }
% \caption{Illustration of the three canonical SWAN operating modes in the proposed random access protocol.}
% \label{Figure_PAN_Protocol}
% \vspace{-15pt}
% \end{figure}
{The controllability and implementation complexity of SWAN depend on how the segment feed points are connected to the RF front-end and how the elementary PA configurations are scheduled across segments. We consider three canonical operating modes, namely OS, SA, and SM, which will later be used by different stages of the proposed access protocol.}

\subsubsection{One-Segment Selection (OS)}\label{subsubsec_SS}

In OS mode, the AP connects only one segment to a single RF chain in a given slot. Let $\bar m\in\mathcal M$ be the selected segment and let $p_{\bar m}\in\mathcal P_{\bar m}$ be its activated PA location. The effective uplink channel then reduces to the single coefficient $\zeta_k^{\mathrm{ul}}(\bar m,p_{\bar m})$, and the received signal at the AP can be written as follows:
\begin{align}\label{eq_SS_ul}
y_k^{\mathrm{ul}} = \zeta_k^{\mathrm{ul}}(\bar m,p_{\bar m})\sqrt{\rho_k}s_k+n^{\mathrm{ul}},
\end{align}
where $s_k\sim\mathcal{CN}(0,1)$ is the normalized transmit symbol, $\rho_k$ is the transmit power of the $k$th UE, and $n^{\mathrm{ul}}\sim\mathcal{CN}(0,\sigma^2)$. OS has the lowest hardware complexity, since it only requires a single RF chain and a switching mechanism.

\subsubsection{Segment Aggregation (SA)}\label{subsubsec_SA}

In SA mode, the AP connects a selected group of segments to a single RF chain through an analog combiner. Let $S$ denote the nominal aggregation size and define $G\triangleq\lceil M/S\rceil$ as the number of SA groups. With $\mathcal G\triangleq\{1,\ldots,G\}$, the segment index set of the $g$th group is defined as follows:
\begin{align}\label{eq_group_def}
\mathcal S_g\triangleq\{(g-1)S+1,\ldots,\min(gS,M)\},\quad g\in\mathcal G.
\end{align}
In a slot assigned to the $g$th group, only the segments in $\mathcal S_g$ are activated and connected to the analog combiner, while all other segments remain inactive. Let $p_m\in\mathcal P_m$ denote the activated PA on the $m$th segment, $m\in\mathcal S_g$, and let ${\mathbf p}_g\triangleq[p_m]_{m\in\mathcal S_g}$ collect the group-wise PA configuration. The effective channel of the $k$th UE in this slot is
\begin{align}\label{eq_SA_g_def}
g_{k,g}({\mathbf p}_g)\triangleq\sum_{m\in\mathcal S_g}\zeta_k^{\mathrm{ul}}(m,p_m).
\end{align}
Accordingly, the uplink received signal is given by
\begin{align}\label{eq_SA_ul}
y_{k,g}^{\mathrm{ul}}
=
\sqrt{\rho_k}\,g_{k,g}({\mathbf p}_g)s_k+n_g^{\mathrm{ul}},
\end{align}
where
\begin{align}
n_g^{\mathrm{ul}}\triangleq \sum_{m\in\mathcal S_g}n_m^{\mathrm{ul}}\sim\mathcal{CN}(0,|\mathcal S_g|\sigma^2).
\end{align}
SA retains a single-RF architecture while allowing group-wise analog aggregation across multiple segments.
% and therefore offers a moderate-complexity operating point between SS and SM.

\subsubsection{Segment Multiplexing (SM)}\label{subsubsec_SM}

In SM mode, each active segment is connected to a dedicated RF chain, which enables digital processing across the activated segments. Let $\mathcal S\subseteq\mathcal M$ denote the active segment set in a generic slot, with $R_{\mathcal S}=|\mathcal S|$. For the $m$th active segment, $m\in\mathcal S$, let $p_m\in\mathcal P_m$ denote the activated PA, and let ${\mathbf p}_{\mathcal S}\triangleq[p_m]_{m\in\mathcal S}$ collect the corresponding PA configuration. The effective uplink channel vector ${\mathbf h}_{k,\mathcal S}^{\mathrm{ul}}({\mathbf p}_{\mathcal S})\in\mathbb C^{R_{\mathcal S}\times 1}$ is defined element-wise as follows:
\begin{align}\label{eq_SM_h}
\left[{\mathbf h}_{k,\mathcal S}^{\mathrm{ul}}({\mathbf p}_{\mathcal S})\right]_r
\triangleq \zeta_k^{\mathrm{ul}}(m_r,p_{m_r}),\quad r=1,\ldots,R_{\mathcal S},
\end{align}
where $\{m_1,\ldots,m_{R_{\mathcal S}}\}=\mathcal S$ follows a fixed ordering.
The corresponding uplink received signal is
\begin{align}\label{eq_SM_ul}
{\mathbf y}_{k,\mathcal S}^{\mathrm{ul}}
=
{\mathbf h}_{k,\mathcal S}^{\mathrm{ul}}({\mathbf p}_{\mathcal S})\sqrt{\rho_k}s_k+{\mathbf n}_{\mathcal S}^{\mathrm{ul}},
\end{align}
where ${\mathbf n}_{\mathcal S}^{\mathrm{ul}}\sim\mathcal{CN}({\mathbf 0},\sigma^2{\mathbf I}_{R_{\mathcal S}})$. SM provides the highest degree of slot-wise controllability, but it also incurs higher RF chain and baseband complexity.

{Overall, SWAN offers two key advantages for uplink random access. First, its segmented implementation suppresses the IAR effect and leads to a tractable uplink signal model. Second, slot-dependent activation of different segments or segment groups over the long waveguide aperture induces distinct configuration-dependent channel responses, which reduces the correlation among channel oracle and access codewords and thereby alleviates collision probability.}
\vspace{-5pt}
\subsection{Proposed SWAN-enabled Access Protocol}\label{subsec_PASS_Access_framework}
{We now present the proposed SWAN-enabled random-access protocol. The protocol comprises two coupled modules, namely a channel-oracle module and an access module. Its main principle is to use the controllable SWAN configuration space as a protocol-level resource. Specifically, the AP activates a subset of SWAN configurations, while each UE learns the configuration-dependent channel variations and uses the resulting information to evaluate candidate access slots.}

{The protocol is slot-synchronous. At the beginning of each oracle or access slot, the AP applies one prescribed SWAN configuration from a pre-broadcast codebook, which incurs a switching interval of duration $T_{\mathrm{sw}}$. Since SWAN configuration control is centralized at the AP whereas channel observations are acquired locally at the UEs, the AP determines the configuration schedule, and each UE uses its local observations for channel prediction and access-slot evaluation.}

\subsubsection{Channel Oracle}
{The channel-oracle module acquires the configuration-dependent channel information required for access with limited overhead. The AP operates SWAN in the OS mode and broadcasts pilots over a sparse set of oracle configurations. From these observations, each UE estimates a low-dimensional geometric state and reconstructs the channel responses over $\mathcal C$. The learned channel responses are then used by the access module to evaluate candidate access slots.}

\subsubsection{Access}
{The access module uses the learned channel responses to coordinate uplink transmissions over $N_{\mathrm{ac}}$ access slots. In conventional random access, UEs contend for slots without knowing how the slot-dependent propagation configuration affects their received signal quality. In contrast, the proposed protocol lets each active UE evaluate the candidate slots using its learned channel responses and then perform oracle-guided access. Without the oracle output, this selective slot access would be unavailable, and the protocol would reduce to conventional contention without exploiting SWAN reconfigurability.}

{This two-module structure links AP-controlled propagation, UE-side channel prediction, and uplink access decision making. We next define the performance metrics used to quantify both the access gain and the overhead required to acquire the oracle.}

\subsubsection{Performance Metrics}\label{subsubsec_performance_metrics}

{We consider one protocol period composed of a channel-oracle stage followed by an access stage. The access stage contains the previously introduced $N_{\mathrm{ac}}$ access slots, and each SWAN configuration update incurs the switching interval $T_{\mathrm{sw}}$. Let $N_{\mathrm{co}}$ denote the number of oracle slots, $T_{\mathrm s}$ the useful transmission duration of each access slot, $L_{\mathrm{co}}$ the pilot length of each oracle slot, and $T_{\mathrm{symb}}$ the symbol duration. The duration of the access stage is given by}
\begin{align}\label{eq_Tac_def}
T_{\mathrm{ac}}=N_{\mathrm{ac}}\big(T_{\mathrm s}+T_{\mathrm{sw}}\big),
\end{align}
{whereas that of the oracle stage is}
\begin{align}\label{eq_Tco_metric}
T_{\mathrm{co}}=N_{\mathrm{co}}\big(L_{\mathrm{co}}T_{\mathrm{symb}}+T_{\mathrm{sw}}\big).
\end{align}
{To account for the effective cost of oracle acquisition, we introduce a tunable oracle-penalty factor $\alpha\ge 0$ and define the overall protocol duration as follows:}
\begin{align}\label{eq_Tperiod_def}
T_{\mathrm p}=T_{\mathrm{ac}}+\alpha T_{\mathrm{co}}.
\end{align}

{Let $K$ denote the number of active UEs attempting access in one protocol period, and let $K_{\mathrm a}$ denote the number of UEs that successfully complete access. The expected access success probability is defined as follows:}
\begin{align}\label{eq_Pac_bar_def}
\bar P_{\mathrm{ac}}\triangleq\frac{\mathbb E[K_{\mathrm a}]}{K}.
\end{align}
{Furthermore, assuming that each successfully accessed UE contributes an effective payload duration of $T_{\mathrm F}$, we define the expected access-stage throughput as follows:}
\begin{align}\label{eq_TP_ac_def}
\mathrm{TP}_{\mathrm{ac}} \triangleq \frac{\mathbb E[K_{\mathrm a}]T_{\mathrm F}}{T_{\mathrm{ac}}} = \frac{\mathbb E[K_{\mathrm a}]T_{\mathrm F}}{N_{\mathrm{ac}}\big(T_{\mathrm s}+T_{\mathrm{sw}}\big)},
\end{align}
{and the expected overall throughput as follows:}
\begin{align}\label{eq_TP_def}
\mathrm{TP} \!\triangleq\! \frac{\mathbb E[K_{\mathrm a}]T_{\mathrm F}}{T_{\mathrm p}} \!=\! \frac{\mathbb E[K_{\mathrm a}]T_{\mathrm F}}{N_{\mathrm{ac}}\big(T_{\mathrm s}\!+\!T_{\mathrm{sw}}\big)\!+\!\alpha N_{\mathrm{co}}\big(L_{\mathrm{co}}T_{\mathrm{symb}}\!+\!T_{\mathrm{sw}}\big)}.
\end{align}
{Here, $\mathrm{TP}_{\mathrm{ac}}$ isolates the efficiency of the access stage itself, whereas $\mathrm{TP}$ further captures the tradeoff between the oracle gain and the associated oracle overhead. Therefore, when the oracle stage is fixed and identical across the compared schemes, we adopt $\mathrm{TP}_{\mathrm{ac}}$ to focus on the access module behavior. Otherwise, we adopt $\mathrm{TP}$ to evaluate the overall protocol efficiency.}

\section{The Channel Oracle Module}\label{sec_SWAN_SS_oracle}
This section details the channel oracle module utilizing the OS mode. The AP transmits pilots over a reduced set of OS configurations. By leveraging these limited observations, each UE can predict the complete uplink channel variations across all candidate PA configurations.
% including those configurations that are not explicitly sampled.

\subsection{Channel Oracle Codebook Formulation}\label{subsec_oracle_formulation}
{Under the slot-synchronous protocol described in Section~\ref{sec_SWAN_protocols}, the AP updates one OS configuration at the beginning of each oracle slot, and pilot observations are collected after the switching interval. We further assume block fading within the oracle period: for the $k$th UE, the channel response $\zeta_k^{\mathrm{dl}}(m,p)$ remains constant within each slot and is quasi-static over the $N_{\mathrm{co}}$ oracle slots\footnote{This standard training assumption can be relaxed by shortening the oracle period or by incorporating parameter-tracking mechanisms.}.}
{Thus, the collected pilots can be viewed as sparse samples of the same configuration-dependent channel variation. Since the oracle codebook does not exhaust all PA configurations, the unobserved responses must be reconstructed from the structural relation among PA locations. Under the SWAN-based LoS channel model, this relation is determined by the UE location and the PA coordinate along each segment. }

{During the oracle stage, the OS mode activates one discrete PA configuration $(m,p)$ per slot, where the physical PA coordinate is $\psi_p^m$. Since the UEs lie on the $xy$-plane, the horizontal coordinate of the $k$th UE is denoted by $\boldsymbol\theta_k\triangleq[u_{x,k},u_{y,k}]^{\mathsf T}$ for notational simplicity.} Based on~\eqref{eq_ho_def} and~\eqref{eq_hi_def}, the downlink channel coefficient under PA configuration $(m,p)$ is parameterized as follows:
\begin{align}\label{eq_zeta_parametric_config}
\zeta_k^{\mathrm{dl}}(m,p;\boldsymbol\theta_k)
=h_{\mathrm i}({\bm\psi}_p^m,{\bm\psi}_0^m)\,h_{\mathrm o}(\boldsymbol\theta_k,\psi_p^m).
\end{align}
Here, the free-space channel propagation component $h_{\mathrm o}(\boldsymbol\theta_k,\psi_p^m)$ depends solely on the geometric distance
\begin{subequations}
\begin{align}\label{eq_ho_parametric}
h_{\mathrm o}(\boldsymbol\theta_k,\psi_p^m) &= \eta^{\frac12}\frac{{\rm e}^{-{\rm j} k_0 r_{k,m,p}}}{r_{k,m,p}},\\ \label{eq_rkmp_def}
r_{k,m,p} &\triangleq \sqrt{(u_{x,k}-\psi_p^m)^2+(u_{y,k}-\psi_{\mathrm w})^2+h^2}.
\end{align}
\end{subequations}
% Furthermore, the in-waveguide propagation component $h_{\mathrm i}(x,\psi_0^m)$ is expressed as follows:
% \begin{align}\label{eq_hi_parametric}
% h_{\mathrm i}(x,\psi_0^m) = 10^{-\frac{\kappa}{20}|x-\psi_0^m|}{\rm e}^{-{\rm j}\frac{2\pi}{\lambda_{\mathrm g}}|x-\psi_0^m|}.
% \end{align}
{Since $h_{\mathrm i}({\bm\psi}_p^m,{\bm\psi}_0^m)$ is determined by the waveguide structure and is therefore known for all UEs, the oracle does not need to estimate an independent complex coefficient for every PA configuration. Instead, each UE only estimates its own geometric coordinate $\boldsymbol\theta_k$, from which its full configuration-dependent channel response can be analytically reconstructed.}

Motivated by this, we now design the oracle codebook. To minimize the training overhead, the OS-based oracle codebook selects only a small subset of $Q_{\mathrm{co}} \ll P$ sampled PA indices for the $m$th segment. This subset is denoted by
\begin{align}
\mathcal P_m^{\mathrm{co}}\subseteq\mathcal P_m,\qquad |\mathcal P_m^{\mathrm{co}}|=Q_{\mathrm{co}}.
\end{align}
% The corresponding PA coordinates are given by $\{\psi_p^m:p\in\mathcal P_m^{\mathrm{co}}\}$.
To establish a simple benchmark for codebook evaluation, a straightforward approach is to select PA configurations that are evenly distributed along the segment. Specifically, the uniform-index baseline is given by
\begin{align}\label{eq_Pco_uniform_index}
\mathcal{P}_m^{\mathrm{co}} = \left\{1+\left\lfloor\frac{(i-1)(P-1)}{Q_{\mathrm{co}}-1}\right\rfloor:\ i=1,\ldots,Q_{\mathrm{co}}\right\}.
\end{align}
Following this segment-wise selection, the global oracle codebook across all $M$ segments is formulated as follows:
\begin{align}\label{eq_Phi_co_def}
\Phi_{\mathrm{co}} \triangleq \{(m,p):m\in\mathcal M,\ p\in\mathcal P_m^{\mathrm{co}}\},
\end{align}
yielding a total training overhead of
\begin{align}\label{eq_Nco_MQ}
N_{\mathrm{co}}=|\Phi_{\mathrm{co}}|=\sum\nolimits_{m=1}^{M}|\mathcal P_m^{\mathrm{co}}|=MQ_{\mathrm{co}}.
\end{align}
While the uniform subset in \eqref{eq_Pco_uniform_index} serves as an intuitive method, it does not guarantee optimal oracle performance. We next strategically design these $Q_{\mathrm{co}}$ sampled PA configurations to maximize the Fisher information for estimating $\boldsymbol\theta_k$, thereby improving the accuracy of channel reconstruction.
\vspace{-5pt}
\subsection{ML-Based Uplink Channel Reconstruction}\label{subsec_oracle_estimation}

{We next formulate the pilot observation model and derive the corresponding maximum-likelihood (ML) estimator, which provide the statistical basis for Fisher information-based oracle design.} For the $n$th oracle slot, let $c[n]=(m[n],p[n])\in\Phi_{\mathrm{co}}$ denote the selected discrete PA configuration. The received pilot vector at the UE is given by
\begin{align}\label{eq_wk_oracle}
\mathbf w_k[n] = \sqrt{\rho_a}\,\zeta_k^{\mathrm{dl}}(m[n],p[n];\boldsymbol\theta_k)\,\mathbf v_{\mathrm{co}}+\boldsymbol\eta_k[n],
\end{align}
where $\mathbf v_{\mathrm{co}}\in {\mathbb C}^{L_{\mathrm{co}}\times 1}$ is the deterministic pilot sequence with $\mathbb E\{\|\mathbf v_{\mathrm{co}}\|_2^2\}=L_{\mathrm{co}}$, and $\boldsymbol\eta_k[n]\sim\mathcal{CN}(\mathbf 0,\sigma^2\mathbf I_{L_{\mathrm{co}}})$ denotes the additive white Gaussian noise (AWGN).

Since the transmitted pilot waveform is known at the UE, we apply matched filtering to obtain the sufficient statistic. This operation projects the $L_{\mathrm{co}}$-dimensional received vector into a single complex scalar, which preserves all information regarding the geometric parameters as follows:
\begin{align}\label{eq_yq_sufficient}
y_k[n] & \triangleq \frac{1}{L_{\mathrm{co}}}\mathbf v_{\mathrm{co}}^{\mathsf H}\mathbf w_k[n]\notag \\
& = \sqrt{\rho_a}\,\zeta_k^{\mathrm{dl}}(m[n],p[n];\boldsymbol\theta_k)+\nu_k[n].
\end{align}
Here, the equalized noise is distributed as follows:
\begin{align}\label{eq_nu_distribution}
\nu_k[n]\sim\mathcal{CN}\!\left(0,\frac{\sigma^2}{L_{\mathrm{co}}}\right),
\end{align}
which is independent across different slots $n$. 
% The variance of the projected noise $\nu_k[n]$ is:
% \begin{align}
% \mathrm{Var}(\nu_k[n]) &=\frac{1}{L_{\mathrm{co}}^2}\mathbf v_{\mathrm{co}}^{\mathsf H}\mathbb E\{\boldsymbol\eta_k[n]\boldsymbol\eta_k[n]^{\mathsf H}\}\mathbf v_{\mathrm{co}}\notag\\
% &=\frac{\sigma^2}{L_{\mathrm{co}}^2}\|\mathbf v_{\mathrm{co}}\|_2^2 =\frac{\sigma^2}{L_{\mathrm{co}}}.
% \end{align}
% Equation~\eqref{eq_yq_sufficient} demonstrates that each oracle slot yields a scalar noisy observation centered around the true geometric channel coefficient. Consequently, the complete oracle dataset is characterized by the collection $\{y_k[n]\}_{n=1}^{N_{\mathrm{co}}}$. Furthermore, increasing the number of sampled PA locations $Q_{\mathrm{co}}$ enhances the geometric diversity of these observations, which is essential for accurate parameter identifiability.
Thus, the oracle dataset of the $k$th UE is given by the scalar sufficient statistics $\{y_k[n]\}_{n=1}^{N_{\mathrm{co}}}$, where increasing $Q_{\mathrm{co}}$ provides greater geometric diversity for identifying $\boldsymbol\theta_k$.
% Given the sufficient statistic in \eqref{eq_yq_sufficient}, we proceed to formulate the location estimator. Specifically, the negative log-likelihood function of $\boldsymbol\theta$, dropping the additive constants, can be written as follows:
% Based on the sufficient statistics $\{y_k[n]\}_{n=1}^{N_{\mathrm{co}}}$, 
For a trial coordinate $\boldsymbol\theta\in\mathcal U$, the negative log-likelihood function, after dropping additive constants, is given by
\begin{align}
\mathcal L_k(\boldsymbol\theta)=\frac{L_{\mathrm{co}}}{\sigma^2}\sum_{n=1}^{N_{\mathrm{co}}}\left|y_k[n]-\sqrt{\rho_a}\,\zeta^{\mathrm{dl}}(m[n],p[n];\boldsymbol\theta)\right|^2.
\end{align}
Therefore, the ML estimate of $\boldsymbol\theta_k$ is obtained by solving the following nonlinear least-squares problem
\begin{align}\label{eq_theta_ml}
\hat{\boldsymbol\theta}_k
=\arg\min_{\boldsymbol\theta\in\mathcal U}
\sum_{n=1}^{N_{\mathrm{co}}}
\left|y_k[n]-\sqrt{\rho_a}\,\zeta^{\mathrm{dl}}(m[n],p[n];\boldsymbol\theta)\right|^2.
\end{align}

% The estimate $\hat{\boldsymbol\theta}_k$ is an intermediate geometric estimate, while the oracle output is the reconstructed channel table over $\mathcal C$. 
% With $\hat{\boldsymbol\theta}_k$, the $k$th UE can analytically reconstruct the complete downlink channel map for every possible configuration in the space $\mathcal C$
After obtaining $\hat{\boldsymbol\theta}_k$, the $k$th UE analytically reconstructs the downlink channel responses over $\mathcal C$ as follows:
\begin{align}\label{eq_model_based_reconstruction}
\hat\zeta_k^{\mathrm{dl}}(m,p) = h_{\mathrm i}({\bm\psi}_p^m,{\bm\psi}_0^m)\,h_{\mathrm o}(\hat{\boldsymbol\theta}_k,\psi_p^m), \quad \forall (m,p)\in\mathcal C.
\end{align}
{Thus, the proposed oracle module converts sparse pilot observations into complete configuration-dependent channel responses through low-dimensional geometric estimation. Finally, assuming channel reciprocity, the corresponding uplink channel coefficients required for the subsequent access module are readily inferred as follows:}
\begin{align}\label{eq_ul_infer}
\hat\zeta_k^{\mathrm{ul}}(m,p)=\hat\zeta_k^{\mathrm{dl}}(m,p),\quad \forall (m,p)\in\mathcal C.
\end{align}
% In summary, the system successfully predicts the global access codebook from a highly limited set of oracle measurements through this physically interpretable framework.
\vspace{-10pt}
\subsection{FIM-Guided Oracle Codebook Design}\label{subsec_oracle_crlb}
{We now quantify the oracle estimation accuracy and derive a principled rule for choosing the oracle sampling size $Q_{\mathrm{co}}$. The purpose is to identify how the sampled PA configurations affect the estimation of $\boldsymbol\theta_k$ as well as the channel reconstruction accuracy.} Specifically, for a trial coordinate $\boldsymbol\theta=[u_x,u_y]^{\mathsf T}$, define the sample mean function as follows:
\begin{align}
\mu_n(\boldsymbol\theta)\triangleq\sqrt{\rho_a}\,\zeta^{\mathrm{dl}}(m[n],p[n];\boldsymbol\theta),
\end{align}
the corresponding gradient vector is given by 
\begin{align}
\mathbf g_n(\boldsymbol\theta)\triangleq\nabla_{\boldsymbol\theta}\zeta^{\mathrm{dl}}(m[n],p[n];\boldsymbol\theta)\in\mathbb C^{2\times 1}.
\end{align}
Assuming circularly symmetric complex Gaussian noise with variance $\sigma^2/L_{\mathrm{co}}$, the FIM for estimating $\boldsymbol\theta$ is then given by
\begin{align}\label{eq_FIM_main}
\mathbf J(\boldsymbol\theta) = \frac{2\rho_a L_{\mathrm{co}}}{\sigma^2} \sum_{n=1}^{N_{\mathrm{co}}} \Re\!\left\{\mathbf g_n(\boldsymbol\theta)\mathbf g^{\mathsf H}_n(\boldsymbol\theta)\right\}.
\end{align}
{Equation~\eqref{eq_FIM_main} shows that the oracle accuracy is governed by two factors: the pilot observation quality captured by $\rho_a L_{\mathrm{co}}/\sigma^2$, and the geometric information contributed by the sampled PA configurations in the summation. In the considered protocol, $L_{\mathrm{co}}$ is fixed by the oracle-slot structure and is not treated as a codebook design variable. Hence, the oracle design focuses on the number and placement of sampled configurations, namely $N_{\mathrm{co}}=MQ_{\mathrm{co}}$ and the subset $\mathcal P_m^{\mathrm{co}}$ on each segment.} Consequently, the covariance of the location estimate is lower-bounded by the inverse FIM as follows:
\begin{align}\label{eq_crb_theta}
\mathrm{Cov}(\hat{\boldsymbol\theta}_k)\succeq \mathbf J(\boldsymbol\theta_k)^{-1}.
\end{align}

To evaluate \eqref{eq_FIM_main}, we must compute the closed-form gradients. For a generic PA configuration $(m,p)$, define the geometric distance
\begin{align}
r_{m,p}(\boldsymbol\theta)\triangleq\sqrt{\Delta_{x,m,p}^2+\Delta_y^2+d^2},
\end{align} 
where $\Delta_{x,m,p}\triangleq u_x-\psi_p^m$ and $\Delta_y \triangleq u_y-\psi_{\mathrm w}$. Recall that $\zeta=h_{\mathrm i}h_{\mathrm o}$. Since the in-waveguide component $h_{\mathrm i}$ is independent of the user location $\boldsymbol\theta$, the gradient simplifies to
\begin{align}
\nabla_{\boldsymbol\theta}\zeta=h_{\mathrm i}\nabla_{\boldsymbol\theta}h_{\mathrm o}.
\end{align}
Based on the free-space channel model $h_{\mathrm o}=\eta^{\frac{1}{2}}\frac{{\rm e}^{-{\rm j}k_0r_{m,p}}}{r_{m,p}}$, the derivative with respect to (w.r.t.) the distance $r_{m,p}$ is
% \begin{align}
% \frac{\partial h_{\mathrm o}}{\partial r_{m,p}}
% &= \eta^{\frac{1}{2}}{\rm e}^{-{\rm j}k_0r_{m,p}}
% \left(-\frac{1}{r_{m,p}^2}-\frac{{\rm j}k_0}{r_{m,p}}\right)= -h_{\mathrm o}\left(\frac{1+{\rm j}k_0r_{m,p}}{r_{m,p}}\right).
% \end{align}
\begin{align}
\frac{\partial h_{\mathrm o}}{\partial r_{m,p}} = -h_{\mathrm o}\left(\frac{1+{\rm j}k_0r_{m,p}}{r_{m,p}}\right).
\end{align}
The partial derivatives of $r_{m,p}$ w.r.t. the user coordinates are
\begin{align}
\frac{\partial r_{m,p}}{\partial u_x}=\frac{\Delta_{x,m,p}}{r_{m,p}}, \qquad \frac{\partial r_{m,p}}{\partial u_y}=\frac{\Delta_y}{r_{m,p}}.
\end{align}
Applying the chain rule, we obtain the partial derivatives of $h_{\mathrm o}$, which can be written as follows:
% \begin{align}\label{eq_dhoux}
% \frac{\partial h_{\mathrm o}}{\partial u_x} &= \frac{\partial h_{\mathrm o}}{\partial r_{m,p}}\frac{\partial r_{m,p}}{\partial u_x} = -h_{\mathrm o}\left(\frac{\Delta_{x,m,p}}{r_{m,p}^2}\right)(1+{\rm j}k_0r_{m,p}),\\ \label{eq_dhouy}
% \frac{\partial h_{\mathrm o}}{\partial u_y} &= \frac{\partial h_{\mathrm o}}{\partial r_{m,p}}\frac{\partial r_{m,p}}{\partial u_y} =-h_{\mathrm o}\left(\frac{\Delta_y}{r_{m,p}^2}\right)(1+{\rm j}k_0r_{m,p}).
% \end{align}
\begin{align}\label{eq_dhoux}
    \frac{\partial h_{\mathrm o}}{\partial u_x} &= \frac{\partial h_{\mathrm o}}{\partial r_{m,p}}\frac{\partial r_{m,p}}{\partial u_x} = -h_{\mathrm o}\left(\frac{\Delta_{x,m,p}}{r_{m,p}^2}\right)(1+{\rm j}k_0r_{m,p}),
\end{align}
\begin{align}\label{eq_dhouy}
    \frac{\partial h_{\mathrm o}}{\partial u_y} &= \frac{\partial h_{\mathrm o}}{\partial r_{m,p}}\frac{\partial r_{m,p}}{\partial u_y} =-h_{\mathrm o}\left(\frac{\Delta_y}{r_{m,p}^2}\right)(1+{\rm j}k_0r_{m,p}).
\end{align}
Therefore, the gradients of the overall channel coefficient $\zeta$ are concisely expressed as follows:
\begin{subequations}\label{eq_dzeta_theta}
\begin{align}
\frac{\partial \zeta}{\partial u_x} &= -\zeta\left(\frac{\Delta_{x,m,p}}{r_{m,p}^2}\right)(1+{\rm j}k_0r_{m,p}),  \\
\frac{\partial \zeta}{\partial u_y} &= -\zeta\left(\frac{\Delta_y}{r_{m,p}^2}\right)(1+{\rm j}k_0r_{m,p}).
\end{align}
\end{subequations}
Substituting \eqref{eq_dzeta_theta} back into \eqref{eq_FIM_main} yields an explicit and computable FIM for any candidate oracle codebook.

While the CRB bounds the parameter estimation error, the ultimate performance of the access phase depends strictly on the channel reconstruction error. Therefore, we must translate the parameter-space CRB into the channel-domain mean squared error (MSE). By linearizing the channel function around the true location $\boldsymbol\theta_k$, we have
\begin{align}
\zeta^{\mathrm{dl}}(m,p;\hat{\boldsymbol\theta}_k)-\zeta^{\mathrm{dl}}(m,p;\boldsymbol\theta_k) \approx \nabla_{\boldsymbol\theta}\zeta^{\mathrm{dl}}(m,p;\boldsymbol\theta_k)^{\mathsf T}(\hat{\boldsymbol\theta}_k-\boldsymbol\theta_k).
\end{align}
This linear approximation allows us to bound the reconstruction error variance as follows:
\begin{align}\label{eq_mse_transfer}
&\mathbb E\!\left[\left|\zeta^{\mathrm{dl}}(m,p;\hat{\boldsymbol\theta}_k)
-\zeta^{\mathrm{dl}}(m,p;\boldsymbol\theta_k)\right|^2\right]\notag\\
&\quad \lesssim \nabla\zeta(m,p)^{\mathsf H}
\,\mathrm{Cov}(\hat{\boldsymbol\theta}_k)\,\nabla\zeta(m,p),
\end{align}
and by applying the inequality from \eqref{eq_crb_theta}, the channel MSE at any PA configuration $(m,p)$ is bounded by
\begin{align}\label{eq_mse_crb_bound}
\mathrm{MSE}_{\zeta}(m,p)
\lesssim \nabla\zeta(m,p)^{\mathsf H}
\mathbf J(\boldsymbol\theta_k)^{-1}\nabla\zeta(m,p).
\end{align}

This analytical bound enables a systematic selection of the sample size $Q_{\mathrm{co}}$. Given a target channel-error tolerance $\delta_{\mathrm{ch}}^2$, we choose the smallest $Q_{\mathrm{co}}$ that satisfies
\begin{align}\label{eq_Q_rule}
\sup_{\boldsymbol\theta\in\mathcal U} \sup_{(m,p)\in\mathcal C} \nabla\zeta(m,p;\boldsymbol\theta)^{\mathsf H} \mathbf J(\boldsymbol\theta;Q_{\mathrm{co}})^{-1} \nabla\zeta(m,p;\boldsymbol\theta) \le\delta_{\mathrm{ch}}^2.
\end{align}
Equation~\eqref{eq_Q_rule} serves as the fundamental decision rule for the oracle design. It ensures that the chosen sample density guarantees the channel reconstruction accuracy under the worst-case scenario across all possible UE locations and candidate PA configurations.

After determining $Q_{\mathrm{co}}$, the specific spatial selection of these oracle PA configurations can be directly optimized. While uniform sampling provides a practical baseline, it is generally suboptimal due to the nonuniform waveguide attenuation and non-linear near-field geometric variations. A more principled approach is to adopt the D-optimal oracle-point selection criterion as follows:
\begin{align}\label{eq_Dopt}
\max_{\mathcal P_m^{\mathrm{co}}\subseteq\mathcal P_m,\ |\mathcal P_m^{\mathrm{co}}|=Q_{\mathrm{co}}} \ \min_{\boldsymbol\theta\in\mathcal U} \log\det\big(\mathbf J_m(\boldsymbol\theta;\mathcal P_m^{\mathrm{co}})\big),
\end{align}
% Alternatively, one can employ the A-optimal criterion:
% \begin{align}\label{eq_Aopt}
% \min_{\mathcal P_m^{\mathrm{co}}\subseteq\mathcal P_m,\ |\mathcal P_m^{\mathrm{co}}|=Q}
% \ \max_{\boldsymbol\theta\in\mathcal U}
% \mathrm{tr}\big(\mathbf J_m(\boldsymbol\theta;\mathcal P_m^{\mathrm{co}})^{-1}\big).
% \end{align}
This optimization problems is combinatorial and can be efficiently approximated using greedy subset selection algorithms based on near-monotone information gain.

\vspace{-5pt}
\section{The SA-Based Access Module}\label{sec_SWAN_access}

{After the channel-oracle stage, this section develops an SA-based access module. The proposed scheme enables each UE to make oracle-guided access decisions without broadcasting a full access codebook, and its single-RF analog aggregation structure reduces the AP-side energy consumption.}
%The derivation begins with the formulation of the access codebook, followed by detailed signal models and outage analyses, ultimately leading to explicit design rules for optimizing access performance.
\vspace{-5pt}
\subsection{SA-based Access Codebook Formulation}\label{subsec_sa_model}

Building on the group-wise SA model in \eqref{eq_group_def}--\eqref{eq_SA_ul}, the access codebook specifies which PA locations can be selected within each activated group and how the group attempts are scheduled over time. For the $m$th segment, we define a subset of access anchors to discretize the spatial coverage, which is given by
\begin{align}
\mathcal P_m^{\mathrm{ac}}=\{p_m^{(1)},\ldots,p_m^{(Q_{\mathrm{ac}})}\}\subseteq\mathcal P_m.
\end{align}
These $Q_{\mathrm{ac}}\leq P$ anchors are uniformly distributed along the segment with a fixed spacing, characterized by 
\begin{align}\label{eq_dx_def}
\psi_q^m=\psi_0^m+(q-1)\Delta_x,
\end{align}
where $\Delta_x=\tfrac{L}{Q_{\mathrm{ac}}-1}$ and $q\in\{1,\ldots,Q_{\mathrm{ac}}\}$.
{One SA access period consists of $N_{\mathrm{ac}}=GQ_{\mathrm{ac}}$ slots, indexed by the group-attempt pair $(g,q)$, where $g\in\mathcal G$ is the active group index and $q\in\{1,\ldots,Q_{\mathrm{ac}}\}$ is the \emph{attempt index} within the $g$th group. Thus, the group-level schedule is deterministic and assigns $Q_{\mathrm{ac}}$ attempts to each group.}
The linear slot index $t\in\{1,\ldots,N_{\mathrm{ac}}\}$ is related to the group-attempt pair $(g,q)$ by
\begin{subequations}\label{eq_t_gq}
\begin{align}
g(t) &= 1+\left\lfloor\frac{t-1}{Q_{\mathrm{ac}}}\right\rfloor, \\
q(t) &= 1+\big((t-1)\bmod Q_{\mathrm{ac}}\big).
\end{align}
\end{subequations}
{Accordingly, the SA codebook specifies only the group-level schedule. For the $(g,q)$th access slot, the AP activates $\mathcal S_g$ and independently samples one anchor from $\mathcal P_m^{\mathrm{ac}}$ for each $m\in\mathcal S_g$. The access diversity is therefore generated by intra-group PA randomization rather than by a slot-specific access codebook.}

% \subsection{Performance Metrics}\label{subsec_performance_metrics}
\subsection{CLT-Based Outage Probability Analysis}\label{subsec_oracle_projection}
For the $q$th attempt of the $g$th group, let ${\mathbf p}_{g,q}$ denote the PA realization randomly selected on the active segments. Following the SA signal model in \eqref{eq_SA_g_def} and \eqref{eq_SA_ul}, the received signal is
\begin{align}\label{eq_yk_sa}
y_{k,g,q}=\sqrt{\rho_k}\, g_{k,g}({\mathbf p}_{g,q})s_k+n_{k,g,q},
\end{align}
where $n_{k,g,q}\sim\mathcal{CN}(0,|\mathcal S_g|\sigma^2)$. The decoding condition is
\begin{align}\label{eq_snr_sa}
\frac{\rho_k}{|\mathcal S_g|\sigma^2}|g_{k,g}({\mathbf p}_{g,q})|^2\ge\gamma_{\mathrm{ac}},
\end{align}
or equivalently,
\begin{align}
|g_{k,g}({\mathbf p}_{g,q})|^2\ge\Gamma_{k,g},
\end{align}
where
\begin{align}
\Gamma_{k,g}\triangleq\frac{\gamma_{\mathrm{ac}}|\mathcal S_g|\sigma^2}{\rho_k}.
\end{align}
{The oracle stage provides the reconstructed uplink channel responses $\{\hat\zeta_k^{\mathrm{ul}}(m,p):(m,p)\in\mathcal C\}$. The following outage analysis is therefore carried out w.r.t. the reconstructed responses. In one access attempt of the $g$th group, each active segment independently and uniformly samples one anchor from $\mathcal P_m^{\mathrm{ac}}$, i.e.,}
\begin{align}\label{eq_Pk_m_def}
P_{k,m}\sim\mathrm{Unif}(\mathcal P_m^{\mathrm{ac}}),\quad m\in\mathcal S_g.
\end{align}
Define
\begin{align}\label{eq_Xkm_def_new}
X_{k,m}\triangleq\hat\zeta_k^{\mathrm{ul}}(m,P_{k,m}), \quad \hat g_{k,g,q}\triangleq\sum_{m\in\mathcal S_g}X_{k,m}.
\end{align}
Then the single-attempt failure probability conditioned on the reconstructed channel responses is 
\begin{align}\label{eq_Fkg_def}
F_{k,g}(\Gamma_{k,g}) \!\triangleq\!
\Pr\left(|\hat g_{k,g,q}|^2\!<\!\Gamma_{k,g}\!\mid\!
\{\hat\zeta_k^{\mathrm{ul}}(m,p)\}_{(m,p)\in\mathcal C}\!\right).
\end{align}

Since $P_{k,m}$ is uniformly distributed over $\mathcal P_m^{\mathrm{ac}}$, the first two moments of $X_{k,m}$ are obtained by averaging the reconstructed responses over the candidate anchors of the $m$th segment:
\begin{align}\label{eq_mu_km}
\mu_{k,m} &\triangleq \mathbb E[X_{k,m}] = \frac{1}{Q_{\mathrm{ac}}}\sum_{p\in\mathcal P_m^{\mathrm{ac}}}\hat\zeta_k^{\mathrm{ul}}(m,p),
\end{align}
\begin{align} \label{eq_s2_km}
s_{k,m}^2 &\triangleq \mathbb E[|X_{k,m}|^2] =\frac{1}{Q_{\mathrm{ac}}}\sum_{p\in\mathcal P_m^{\mathrm{ac}}}\left|\hat\zeta_k^{\mathrm{ul}}(m,p)\right|^2,
\end{align}
\begin{align} \label{eq_sigma2_km}
\sigma_{k,m}^2 &\triangleq \mathrm{Var}(X_{k,m}) = s_{k,m}^2-|\mu_{k,m}|^2.
\end{align}
Hence, for the $g$th group,
\begin{align}\label{eq_mu_V_group}
\mu_{k,g}=\sum_{m\in\mathcal S_g}\mu_{k,m}, \quad V_{k,g}=\sum_{m\in\mathcal S_g}\sigma_{k,m}^2.
\end{align}
For moderate group size, central-limit-theorem (CLT) gives that
\begin{align}\label{eq_gk_cn}
\hat g_{k,g,q}\approx\mathcal{CN}(\mu_{k,g},V_{k,g}).
\end{align}

\begin{proposition}\label{prop:outage_exact}
Under independent random anchor sampling across segments and attempts, the single-attempt failure probability for the SA-based access module admits the Marcum-$Q$ approximation
\begin{align}\label{eq_Fkg_marcum}
F_{k,g}(\Gamma_{k,g}) \approx 1-Q_1\left(\sqrt{\frac{2|\mu_{k,g}|^2}{V_{k,g}}},\sqrt{\frac{2\Gamma_{k,g}}{V_{k,g}}}\right),
\end{align}
and the outage probability is
\begin{align}\label{eq_out_prod}
P_k^{\mathrm{out}} &\approx \prod_{g=1}^{G}\Big(F_{k,g}(\Gamma_{k,g})\Big)^{Q_{\mathrm{ac}}}\notag \\
&= \prod_{g=1}^{G} \left[1-Q_1\!\left(\sqrt{\frac{2|\mu_{k,g}|^2}{V_{k,g}}},\sqrt{\frac{2\Gamma_{k,g}}{V_{k,g}}}\right)\right]^{Q_{\mathrm{ac}}}.
\end{align}
\end{proposition}
\begin{IEEEproof}
{From \eqref{eq_gk_cn}, $\hat g_{k,g,q}$ is approximated as a complex Gaussian random variable with mean $\mu_{k,g}$ and variance $V_{k,g}$. Hence, $|\hat g_{k,g,q}|$ follows a Rician distribution, whose CDF at $\sqrt{\Gamma_{k,g}}$ gives \eqref{eq_Fkg_marcum}. Since the random anchors are independently selected across attempts and groups, outage occurs only when all $Q_{\mathrm{ac}}$ attempts of every group fail, which yields the product form in \eqref{eq_out_prod}.}
\end{IEEEproof}

Leveraging \eqref{eq_out_prod}, a generic success model for the $k$th UE can be written as follows:
\begin{align}\label{eq_pk_succ_general}
p_k^{\mathrm{succ}}\approx\big(1-P_k^{\mathrm{out}}\big)p_k^{\mathrm{col}},
\end{align}
where $p_k^{\mathrm{col}}$ denotes the conditional probability that the packet of the $k$th UE is not lost due to multiuser contention after at least one access attempt meets the decoding threshold. Therefore,
\begin{align}\label{eq_EKa_sum}
\mathbb E[K_a]=\sum_{k=1}^{K}p_k^{\mathrm{succ}}.
\end{align}
For SA, $N_{\mathrm{ac}}=GQ_{\mathrm{ac}}$ and $N_{\mathrm{co}}=MQ_{\mathrm{co}}$ in \eqref{eq_TP_def}. Therefore, the design variables $(S,Q_{\mathrm{ac}},Q_{\mathrm{co}})$ affect throughput through both the expected number of successful UEs and the protocol duration.

\subsection{Access Tradeoffs and Design Guidelines}\label{subsec_det_guarantees}

\subsubsection{Contention Splitting}
The group size $S$ has two coupled effects in the SA access design. Through $G=\lceil M/S\rceil$, it specifies the number of segment groups available for distributing active UEs across SWAN configurations. Meanwhile, it also specifies the number of segments coherently combined in each access slot, thereby affecting the outage metric in \eqref{eq_out_prod}.

The contention effect can be characterized as follows. Let $K_g$ denote the number of active UEs selecting access slots associated with the $g$th group, where $\sum_{g=1}^{G}K_g=K$. The number of group-level contending UE pairs is given by
\begin{align}\label{eq_pair_collision}
C_{\mathrm{grp}} 
=\sum_{g=1}^{G}\binom{K_g}{2}
=\frac{1}{2}\left(\sum_{g=1}^{G}K_g^2-K\right).
\end{align}
For fixed $K$ and $G$, this quantity is minimized by a balanced group occupancy, which can be written as follows:
\begin{align}\label{eq_pair_collision_balanced}
C_{\mathrm{grp}}\approx\frac{1}{2}\left(\frac{K^2}{G}-K\right).
\end{align}
Therefore, relative to the ungrouped case $G=1$, an oracle-guided access pattern with well-balanced group selections reduces the pairwise contention load by approximately a factor of $G$. This reduction is not guaranteed by the oracle itself; it depends on whether the learned channel responses induce sufficiently diverse group selections across the active UEs.

\subsubsection{Aggregation-Latency Tradeoff}
This contention-splitting gain is accompanied by two design costs. First, since each group is assigned $Q_{\mathrm{ac}}$ attempts, the number of access slots is $N_{\mathrm{ac}}=GQ_{\mathrm{ac}}=\left\lceil\frac{M}{S}\right\rceil Q_{\mathrm{ac}}$.
Thus, decreasing $S$ increases the access duration. Second, a smaller group may reduce the aggregation gain in \eqref{eq_SA_ul}. Conversely, increasing $S$ shortens the group sweep and may improve analog aggregation, but it also concentrates the active UEs over fewer group-associated slots and increases the impact of phase cancellation in the aggregated channel. Hence, $S$ is fundamentally a contention-aggregation tradeoff parameter rather than a purely structural grouping parameter.

% \subsubsection{Throughput-Oriented Parameter Selection}
Overall, heavy access load favors a smaller $S$, which creates more segment groups and alleviates contention, whereas light load or stringent latency constraints favor a larger $S$, which shortens the group sweep and preserves analog aggregation. The parameters $Q_{\mathrm{ac}}$ and $Q_{\mathrm{co}}$ should be enlarged only when their marginal reliability gains, through access diversity and oracle accuracy, respectively, compensate for the additional protocol duration in \eqref{eq_TP_def}. Thus, the SA design should favor moderate grouping rather than either full aggregation or overly fine splitting.

\section{The SM-Based $R$-Access Module}\label{sec_partialSM_access}
This section presents an SM-based access module, referred to as the $R$-access scheme. In each access slot, the AP activates $R$ selected segments and connects them to $R$ dedicated RF chains, thereby enabling spatially multiplexed access with improved coverage and high-load capability.

% \subsection{Slot Configuration and Effective SNR}\label{subsec_partialSM_arch}
As introduced in the SM model in Section~\ref{subsec_SWAN_three_protocols}, the selected segment subset in the $t$th access slot is denoted by
\begin{align}\label{eq_partialSM_active_set}
\mathcal{S}_t \subseteq \mathcal{M},\qquad |\mathcal{S}_t|=R.
\end{align}
For the $m$th active segment, $m\in\mathcal{S}_t$, one PA index $p_m[t]\in\mathcal{P}_m$ is assigned. The resulting slot configuration is
\begin{align}\label{eq_partialSM_slot_config}
c_t \triangleq \big(\mathcal{S}_t,\{p_m[t]\}_{m\in\mathcal{S}_t}\big).
\end{align}
Let $\{m_1,\ldots,m_R\}=\mathcal{S}_t$ follow a fixed ordering. The effective uplink channel vector of the $k$th UE under configuration $c_t$ is defined as follows:
\begin{align}\label{eq_partialSM_h_vec}
\left[\mathbf h_k^{\mathrm{ul}}[t]\right]_r
\triangleq \zeta_k^{\mathrm{ul}}(m_r,p_{m_r}[t]),\quad r=1,\ldots,R .
\end{align}
Since the activated segments are connected to dedicated RF chains, digital combining yields the received SNR
\begin{align}\label{eq_partialSM_snr}
\mathrm{SNR}^{\mathrm{rx}}_{k}[t] = \frac{\rho_k}{\sigma^2}\|\mathbf{h}_k^{\mathrm{ul}}[t]\|_2^2 = \frac{\rho_k}{\sigma^2} \sum_{m\in\mathcal{S}_t}\big|\zeta_k^{\mathrm{ul}}(m,p_m[t])\big|^2.
\end{align}
Equation~\eqref{eq_partialSM_snr} reveals a monotonic SNR accumulation. Specifically, the SNR of each access slot is lower bounded by the contribution of any selected segment. Hence, the subsequent codebook design only needs to ensure that each UE encounters one sufficiently strong segment-anchor pair in at least one access slot.
\vspace{-5pt}
\subsection{$R$-Access Codebook Formulation}\label{subsec_partialSM_codebook}

Unlike the SA scheme, where the intra-group PA configuration is randomized and no full access codebook needs to be broadcast, $R$-access requires a deterministic access codebook shared by the AP and UEs. This codebook assigns a configuration $c_t$ to each access slot and consists of two hierarchical layers: a segment-activation schedule and an intra-segment anchor schedule.

First, the segment-activation schedule determines the subsets $\{\mathcal{S}_t\}$. The $M$ segments are swept by activation blocks of size $R$, and the number of required blocks is $B \triangleq \left\lceil\tfrac{M}{R}\right\rceil$.
For each block index $b\in\{1,\ldots,B\}$, the segment set is constructed as follows:
\begin{align}\label{eq_partialSM_block_set}
\mathcal{S}^{(b)} \triangleq \Big\{(b-1)R+1,\ldots,\min(bR,M)\Big\} \cup \mathcal{W}^{(b)}.
\end{align}
If the last block contains fewer than $R$ segments, $\mathcal{W}^{(b)}$ fills the remaining positions by reusing the first few segment indices, so that every block contains exactly $R$ segments. This construction guarantees that every segment is activated in at least one block, i.e.,
\begin{align}\label{eq_partialSM_cover_property}
\forall m\in\mathcal{M},\ \exists b\in\{1,\ldots,B\}:\ m\in\mathcal{S}^{(b)}.
\end{align}

Second, the anchor schedule determines the PA position $\{p_m[t]\}$. For the $m$th segment, the uniformly distributed access anchors $\{p_m^{(q)}\}_{q=1}^{Q_{\mathrm{ac}}}$ from \eqref{eq_dx_def} are reused. Each access slot is indexed by an anchor index $q$ and a segment-block index $b$, with
\begin{align}\label{eq_partialSM_slot_indexing}
t=(q-1)B+b,\ q\in\{1,\ldots,Q_{\mathrm{ac}}\},\ b\in\{1,\ldots,B\}.
\end{align}
During the $t$th slot, the active segment set is $\mathcal{S}_t=\mathcal{S}^{(b)}$. For each active segment $m\in\mathcal{S}^{(b)}$, the activated PA is set to
\begin{align}\label{eq_partialSM_pin_rule}
p_m[t]=p_m^{(q)}.
\end{align}
The complete $R$-access codebook is then formally defined as follows:
\begin{align}\label{eq_partialSM_codebook_def}
\Phi_{\mathrm{ac}}^{(R)} \triangleq \big\{c_t:\ t=1,\ldots,N_{\mathrm{ac}}^{(R)}\big\},
\end{align}
The resulting access period length is
\begin{align}\label{eq_partialSM_Nac_scaling}
N_{\mathrm{ac}}^{(R)} = BQ_{\mathrm{ac}} = Q_{\mathrm{ac}}\left\lceil\tfrac{M}{R}\right\rceil.
\end{align}
% Given the reconstructed uplink responses $\hat{\zeta}_k^{\mathrm{ul}}(m,p)$ provided by the channel oracle, the $k$th UE applies the broadcasted codebook and selects the access slot with the largest predicted effective channel gain, which is given by
Given the reconstructed uplink responses $\hat{\zeta}_k^{\mathrm{ul}}(m,p)$ provided by the channel oracle, the $k$th UE applies the broadcasted codebook, and the selected access slot is determined by
\begin{align}\label{eq_partialSM_select_rule}
t_k^\star \triangleq \min\arg\max_{t\in\{1,\ldots,N_{\mathrm{ac}}^{(R)}\}}
\sum_{m\in\mathcal{S}_t}\big|\hat{\zeta}_k^{\mathrm{ul}}(m,p_m[t])\big|^2.
\end{align}
In this way, the broadcasted codebook provides a deterministic configuration schedule, while the oracle output enables each UE to identify its most favorable access slot.
% \begin{table*}[!t]
% \centering
% \caption{Comparison of SA and $R$-Access Under Different Hardware Protocols}
% \label{tab:SA_R_access_comparison}
% \renewcommand{\arraystretch}{1.3}
% \setlength{\tabcolsep}{6pt}
% \begin{tabular}{|c|c|c|}
% \hline
%  & \textbf{SA-Based Access} & \textbf{SM-Based $R$-Access} \\ \hline
% RF structure 
% & Single-RF architecture 
% & Multi-RF architecture \\ \hline

% Combining mechanism 
% & Analog coherent summation 
% & Digital energy combining \\ \hline

% Hardware controllability 
% & Limited 
% & High \\ \hline

% Access philosophy 
% & Randomized intra-group access codebook 
% & Deterministic slot-wise access codebook \\ \hline

% Performance gain source 
% & Statistical diversity across codewords 
% & Structured coverage through deterministic codewords \\ \hline

% Reliability nature 
% & Probabilistic 
% & Deterministic \\ \hline
% \end{tabular}
% \end{table*}
\vspace{-10pt}
\subsection{Geometry-Based Coverage Analysis}\label{subsec_partialSM_guarantee}

The monotonic SNR accumulation in \eqref{eq_partialSM_snr} allows the $R$-access codebook to be analyzed through a single strong segment-anchor pair. We therefore first derive a geometry-based lower bound on the gain of at least one such pair. From the propagation model in Section~\ref{sec_SWAN_protocols}, the segment-wise channel gain is
\begin{align}\label{eq_partialSM_zeta_gain}
\big|\zeta_k^{\mathrm{ul}}(m,p)\big|^2 = \big|h_{\mathrm i}(\boldsymbol\psi_p^m,\boldsymbol\psi_0^m)\big|^2 \frac{\eta}{\|\mathbf u_k-\boldsymbol\psi_p^m\|^2}.
\end{align}
With the uniformly spaced access anchors, for any valid UE location $\mathbf u_k\in\mathcal U$, there exists at least one segment-anchor pair $(m^\star,q^\star)$ whose distance is bounded by
\begin{align}\label{eq_partialSM_rcov_ub}
\|\mathbf u_k\!-\!\boldsymbol\psi_{q^\star}^{m^\star}\| \le r_{\mathrm{cov}}^{\mathrm{ub}} \triangleq \sqrt{\left(\frac{L}{2(Q_{\mathrm{ac}}-1)}\right)^2 + Y_{\max}^2 + d^2}.
\end{align}
The in-segment waveguide attenuation is lower bounded by its worst case over length $L$, namely
\begin{align}\label{eq_partialSM_hi_lb}
\big|h_{\mathrm i}(\boldsymbol\psi_{q}^{m},\boldsymbol\psi_0^m)\big|^2 \ge 10^{-\frac{\kappa}{10}L}.
\end{align}
Combining \eqref{eq_partialSM_zeta_gain}, \eqref{eq_partialSM_rcov_ub}, and \eqref{eq_partialSM_hi_lb} gives the following guaranteed gain for at least one segment-anchor pair:
\begin{align}\label{eq_partialSM_Gmin_def}
\big|\zeta_k^{\mathrm{ul}}(m^\star,p_{m^\star}^{(q^\star)})\big|^2\ge G_{\min}^{(R)}\triangleq 10^{-\frac{\kappa}{10}L}\cdot \frac{\eta}{(r_{\mathrm{cov}}^{\mathrm{ub}})^2}.
\end{align}
This bound leads to the following coverage guarantee for the $R$-access codebook. 

\begin{proposition}\label{prop:partialSM_exist_slot}
Consider the $R$-access codebook $\Phi_{\mathrm{ac}}^{(R)}$ formulated in \eqref{eq_partialSM_codebook_def}. For the $k$th UE located at $\mathbf u_k\in\mathcal U$, there exists at least one access slot indexed by $t^\star\in\{1,\ldots,N_{\mathrm{ac}}^{(R)}\}$ satisfying
\begin{align}\label{eq_partialSM_exist_slot}
\|\mathbf h_k^{\mathrm{ul}}[t^\star]\|_2^2 \ge G_{\min}^{(R)}.
\end{align}
Consequently, provided the transmit power satisfies the condition 
\begin{align}\label{eq_partialSM_power_suff}
\frac{\rho_k}{\sigma^2}G_{\min}^{(R)}\ge \gamma_{\mathrm{ac}},
\end{align}
the $k$th UE is guaranteed to experience at least one slot where the received SNR meets the decoding threshold, which eliminates the possibility of a coverage-induced outage.
\end{proposition}
\begin{IEEEproof}
By \eqref{eq_partialSM_Gmin_def}, for any $\mathbf u_k\in\mathcal U$, there exists at least one segment-anchor pair $(m^\star,q^\star)$ such that
\begin{align}
\big|\zeta_k^{\mathrm{ul}}(m^\star,p_{m^\star}^{(q^\star)})\big|^2\ge G_{\min}^{(R)}.
\end{align}
The coverage property in \eqref{eq_partialSM_cover_property} ensures that $m^\star$ belongs to at least one segment block, denoted by $\mathcal S^{(b^\star)}$. Consider the access slot $t^\star=(q^\star-1)B+b^\star$. From the anchor schedule in \eqref{eq_partialSM_pin_rule}, the activated PA on the $m^\star$th segment satisfies $p_{m^\star}[t^\star]=p_{m^\star}^{(q^\star)}$. Since the received channel energy in \eqref{eq_partialSM_snr} is a sum of nonnegative terms, we have
\begin{align}
\|\mathbf h_k^{\mathrm{ul}}[t^\star]\|_2^2
&=
\sum_{m\in\mathcal S^{(b^\star)}}\big|\zeta_k^{\mathrm{ul}}(m,p_m[t^\star])\big|^2 \notag \\
&\ge
\big|\zeta_k^{\mathrm{ul}}(m^\star,p_{m^\star}^{(q^\star)})\big|^2
\ge G_{\min}^{(R)}.
\end{align}
Substituting this bound into \eqref{eq_partialSM_snr} gives $\mathrm{SNR}_{k}^{\mathrm{rx}}[t^\star]\ge \rho_k G_{\min}^{(R)}/\sigma^2$. Hence, \eqref{eq_partialSM_power_suff} guarantees $\mathrm{SNR}_{k}^{\mathrm{rx}}[t^\star]\ge\gamma_{\mathrm{ac}}$, which completes the proof.
\end{IEEEproof}

Proposition~\ref{prop:partialSM_exist_slot} addresses the coverage aspect of $R$-access. In multiuser access, reliability is further governed by slot occupancy. Since each selected slot provides $R$-dimensional observations, a contention-induced failure occurs only when more than $R$ UEs select the same slot.
\vspace{-10pt}
\subsection{Collision Probability Analysis}\label{subsec_partialSM_collision}
Once the oracle output and the codebook are fixed, each UE follows the deterministic rule in \eqref{eq_partialSM_select_rule}. However, the selected slot remains random across the network because UE locations are random. We therefore characterize the resulting contention behavior as follows. Let $K$ active UEs independently draw their locations in the service region and apply the same oracle-guided access rule. The selected slot of a generic UE is a discrete random variable with probability
\begin{align}\label{eq_partialSM_pt_sel}
p_t^{\mathrm{sel}} \triangleq \Pr\!\big(t_k^\star=t\big),\qquad t=1,\ldots,N_{\mathrm{ac}}^{(R)},
\end{align}
where $\sum_t p_t^{\mathrm{sel}}=1$. The quantity $p_t^{\mathrm{sel}}$ is induced jointly by the user-geometry distribution, the deterministic access codebook, and the oracle-guided access rule in \eqref{eq_partialSM_select_rule}.
Let
\begin{align}\label{eq_partialSM_Ct_def}
C_t \triangleq \sum_{k=1}^{K}\mathbf 1\{t_k^\star=t\}
\end{align}
which denotes the number of UEs selecting the $t$th slot. Under independent UE locations, the joint slot-occupancy distribution is given by
\vspace{-5pt}
\begin{align}\label{eq_partialSM_multinomial}
\Pr\big(C_1\!=\!c_1,\ldots,C_{N_{\mathrm{ac}}^{(R)}}\!=\!c_{N_{\mathrm{ac}}^{(R)}}\big)
\!=\!\frac{K!}{\prod_{t=1}^{N_{\mathrm{ac}}^{(R)}}\!\! c_t!}
\!\prod_{t=1}^{N_{\mathrm{ac}}^{(R)}}\! \big(p_t^{\mathrm{sel}}\big)^{c_t},
\end{align}
for any nonnegative integers $\{c_t\}$ satisfying $\sum_{t=1}^{N_{\mathrm{ac}}^{(R)}} c_t=K$. Marginally, one has $C_t\sim\mathrm{Binomial}(K,p_t^{\mathrm{sel}})$.

A collision in $R$-access is declared only when the number of simultaneously selected UEs in one slot exceeds the number of RF chains. Conditioned on the $k$th UE selecting the $t$th slot, its transmission is resolvable if at most $R-1$ other UEs select the same slot. The corresponding conditional resolvability probability can be written as follows: 
\begin{align}\label{eq_partialSM_pcol_cond}
p_{\mathrm{col},t}^{(R)} &\triangleq \Pr(C_t^{(-k)}\le R-1)\notag \\
&= \sum_{n=0}^{R-1}\binom{K-1}{n}\big(p_t^{\mathrm{sel}}\big)^n\big(1-p_t^{\mathrm{sel}}\big)^{K-1-n},
\end{align}
where $C_t^{(-k)}$ counts the other UEs selecting the $t$th slot. Averaging over the selected slot of the $k$th UE yields the collision-resolvability probability, which is given by
\begin{align}\label{eq_partialSM_pcol_avg}
p_{\mathrm{col}}^{(R)}
&= \sum_{t=1}^{N_{\mathrm{ac}}^{(R)}} p_t^{\mathrm{sel}} p_{\mathrm{col},t}^{(R)} \notag\\
&= \sum_{t=1}^{N_{\mathrm{ac}}^{(R)}} p_t^{\mathrm{sel}}
\sum_{n=0}^{R-1}\binom{\!K\!-\!1\!}{n}
\big(p_t^{\mathrm{sel}}\big)^n
\big(1-p_t^{\mathrm{sel}}\big)^{K-1-n}.
\end{align}

\begin{corollary}\label{cor:partialSM_balanced_scaling}
If the oracle-guided slot selection is approximately balanced, i.e., $p_t^{\mathrm{sel}}\simeq 1/N_{\mathrm{ac}}^{(R)}$, then
\begin{align}\label{eq_partialSM_pcol_balanced}
p_{\mathrm{col}}^{(R)}
\simeq
\sum_{n=0}^{R-1}\binom{\!K\!-\!1\!}{n}
\left(\frac{1}{N_{\mathrm{ac}}^{(R)}}\right)^n
\left(1-\frac{1}{N_{\mathrm{ac}}^{(R)}}\right)^{K-1-n}.
\end{align}
For large $M$ with $N_{\mathrm{ac}}^{(R)}\simeq M Q_{\mathrm{ac}}/R$, let $\beta\triangleq (K-1)/(M Q_{\mathrm{ac}})$. Then \eqref{eq_partialSM_pcol_balanced} admits the approximation
\begin{align}\label{eq_partialSM_poisson_scaling}
p_{\mathrm{col}}^{(R)}\approx \Pr\{Z\le R-1\},\qquad Z\sim\mathrm{Poisson}(\beta R).
\end{align}
Hence, increasing $R$ reduces the number of access slots while simultaneously increasing the resolvable occupancy per slot. This scaling identifies $\beta$ as the effective access load. Specifically, $R$-access is favorable when $\beta<1$, whereas $\beta>1$ indicates that increasing $R$ alone cannot compensate for insufficient aggregate access resources.
\end{corollary}

Under the sufficient-power condition in \eqref{eq_partialSM_power_suff}, Proposition~\ref{prop:partialSM_exist_slot} eliminates coverage outage, and the remaining access failure is dominated by the event that more than $R$ UEs select the same slot, i.e., $C_t>R$. Accordingly, \eqref{eq_partialSM_pcol_avg} characterizes the contention-limited reliability of the $R$-access design. When both coverage and contention are considered, the access success probability of the $k$th UE can be approximated by
\begin{align}\label{eq_partialSM_success_general}
p_k^{\mathrm{succ},(R)} \approx p_k^{\mathrm{cov},(R)}\, p_{\mathrm{col}}^{(R)},
\end{align}
where $p_k^{\mathrm{cov},(R)}=1$ under the sufficient-power condition.

The corresponding expected number of successfully admitted UEs is given by
\begin{align}\label{eq_partialSM_EKa}
\mathbb E[K_a^{(R)}]
= \sum_{t=1}^{N_{\mathrm{ac}}^{(R)}} \sum_{n=1}^{R} n\binom{K}{n}\big(p_t^{\mathrm{sel}}\big)^n\big(1-p_t^{\mathrm{sel}}\big)^{K-n},
\end{align}
which is then substituted into the throughput metric in \eqref{eq_TP_def}. Therefore, unlike SA, whose main design variable is the grouping size $S$, $R$-access is shaped by the number of RF chains $R$ and the oracle-guided slot-selection distribution $\{p_t^{\mathrm{sel}}\}$. Increasing $R$ raises the resolvable occupancy of each slot, whereas balancing $\{p_t^{\mathrm{sel}}\}$ across slots reduces overload concentration.
\begin{figure*}[!t]
\centering
\subfigure[Oracle bound versus $Q_{\mathrm{co}}$, where $\rho_a=10$~dBm.]{
\includegraphics[width=0.28\textwidth]{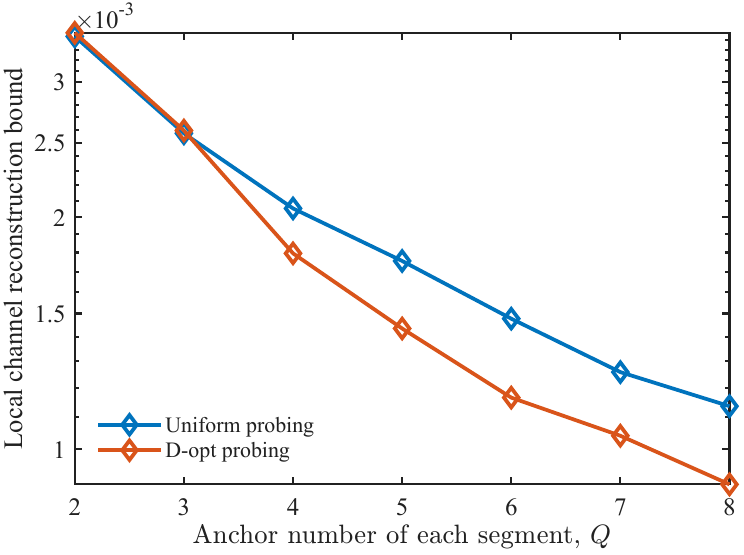}
}
% \hfill
\hspace{0.008\textwidth}
\subfigure[SA outage versus $Q_{\mathrm{ac}}$, where $S=8$, $\rho_k=0$~dBm, and $\gamma_{\mathrm{ac}}=5$~dB.]{
\includegraphics[width=0.28\textwidth]{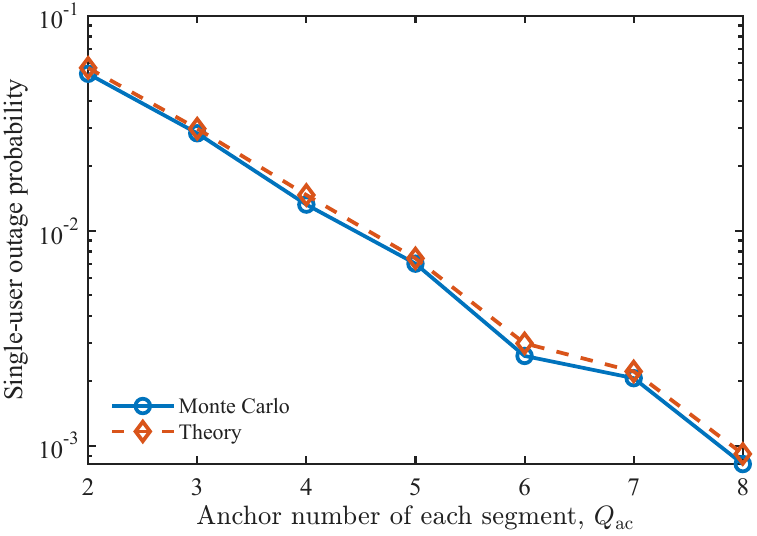}
}
% \hfill
\hspace{0.008\textwidth}
\subfigure[SA outage versus $S$, where $Q_{\mathrm{ac}}=2$.]{
\includegraphics[width=0.28\textwidth]{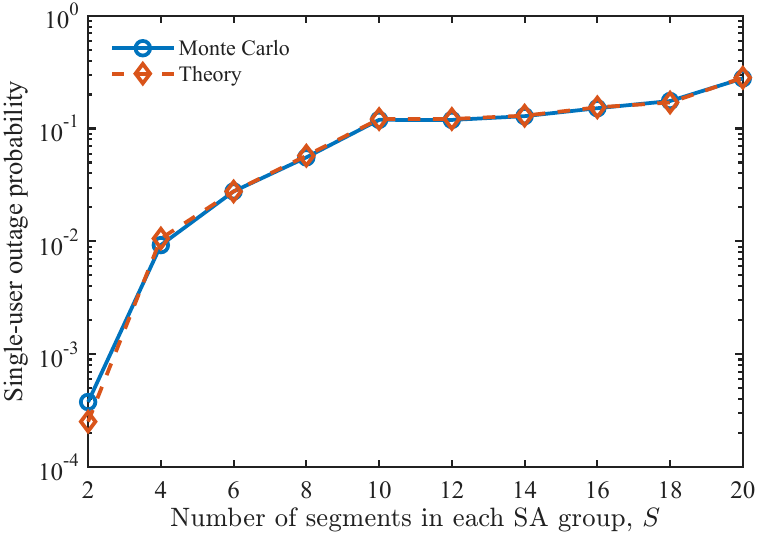}
}
\caption{Oracle quality and single-user SA reliability under the common system configuration.}
\label{fig_oracle_sa_joint}
\vspace{-12pt}
\end{figure*}

\subsection{Design Implications of SA and $R$-Access}
The preceding analysis shows that SA and $R$-access are suited to different operating regimes. In SA, a single RF chain supports an access period of length $N_{\mathrm{ac}}=Q_{\mathrm{ac}}\lceil M/S\rceil$, while the balanced group-level contention scales approximately as $\frac{1}{2}(K^2/G-K)$ with $G=\lceil M/S\rceil$ according to \eqref{eq_pair_collision_balanced}. Hence, SA is attractive when hardware cost, AP transmit power consumption, or implementation simplicity is the primary constraint. In this regime, $S$ should be decreased under heavier access load to create more contention-splitting groups, whereas larger $S$ is preferable under lighter load or stricter latency constraints because it reduces the access duration and preserves analog aggregation.

In contrast, $R$-access moves beyond randomized contention splitting toward spatially resolvable access. By allocating multiple RF chains to each access slot, it shortens the configuration schedule, provides deterministic coverage through \eqref{eq_partialSM_power_suff}, and increases the number of UEs that can be supported in the same slot. Therefore, $R$-access is more suitable for reliability-oriented or latency-sensitive deployments with moderate or high access load.
% provided that the RF chain budget is acceptable. 
Moreover, the effective load factor $\beta$ further shows that multiple RF chains cannot compensate for an insufficient access resource. When $\beta$ is large, the system should enlarge the access codebook or aperture of waveguides rather than only increasing $R$.

\section{Numerical Results}\label{sec_numerical_results}

This section provides representative numerical results for both the oracle and access modules of the proposed framework. The simulations are designed to (i) validate the sampling strategy used in the channel-oracle stage, (ii) verify that the analytical SA model captures the access reliability behavior, and (iii) quantify the architecture-level throughput tradeoff between the single-RF SA scheme and the multi-RF $R$-access scheme.

\subsection{Simulation Setup}

We consider a SWAN deployed over a rectangular service region of size $D_x\times D_y=60~\mathrm{m}\times 10~\mathrm{m}$ with $M=20$ segments and segment length $L=3$~m. The waveguide is placed at $\psi_{\mathrm w}=0$, and the deployment height is set to $h=5$~m. Each segment contains $P=30$ uniformly spaced candidate PA positions. The carrier frequency is $f_{\mathrm c}=30$~GHz, the guided wavelength is $\lambda_{\mathrm g}=\lambda/1.44$, and the in-waveguide attenuation coefficient is $\kappa=0.1$~dB/m. The noise power is fixed at $\sigma^2=-90$~dBm, and UE locations are independently drawn from the uniform distribution over $\mathcal U$.

For the oracle stage, the AP transmits pilots with power $\rho_a=10$~dBm, the oracle pilot length is $L_{\mathrm{co}}=14$, and the oracle sampling size is varied over $Q_{\mathrm{co}}\in\{2,\ldots,8\}$ when the oracle overhead is evaluated. For the access stage, unless otherwise stated, the transmit power of the $k$th UE is $\rho_k=0$~dBm and the decoding threshold is $\gamma_{\mathrm{ac}}=5$~dB. The protocol timing parameters are normalized by the symbol duration and set to $T_{\mathrm{symb}}=1$, $T_{\mathrm s}=20$, $T_{\mathrm{sw}}=0.2$, and $T_{\mathrm F}=T_{\mathrm s}$, respectively. The oracle-overhead penalty in \eqref{eq_Tperiod_def} is $\alpha=1$ for overall throughput evaluations. All Monte Carlo averages are obtained from independently generated UE locations, with the number of realizations specified by the corresponding experiment.

\subsection{Oracle Sampling Accuracy and SA Reliability}
Fig.~\ref{fig_oracle_sa_joint} illustrates the oracle-quality and the single-user SA reliability performance. In Fig.~\ref{fig_oracle_sa_joint}(a), the oracle quality improves monotonically with the oracle sampling size $Q_{\mathrm{co}}$, and the proposed D-optimal sampling design consistently outperforms uniform sampling under the same overhead. Figs.~\ref{fig_oracle_sa_joint}(b) and (c) then examine SA outage for a single UE by assuming accurate reconstructed channel responses, so that oracle error and multiuser collisions are excluded. In Fig.~\ref{fig_oracle_sa_joint}(b), the outage decreases with $Q_{\mathrm{ac}}$, which verifies the outage characterization in \eqref{eq_Fkg_marcum} and \eqref{eq_out_prod}. In Fig.~\ref{fig_oracle_sa_joint}(c), the outage increases with $S$, which indicates that overly large SA groups are undesirable in the single-RF architecture.

These observations lead to the following design implication. Increasing $Q_{\mathrm{co}}$ and optimizing the sampled configurations improve channel predictability, while SA reliability is improved more effectively by anchor densification than by excessive segment aggregation. The reason is that increasing $S$ simultaneously intensifies phase cancellation in the analog coherent sum, raises the effective threshold through $\Gamma_{k,g}=\gamma_{\mathrm{ac}}|\mathcal S_g|\sigma^2/\rho_k$, and reduces the number of contention groups available within one access period. Therefore, the preferred operating region is characterized by geometry-aware oracle sampling, moderate anchor densification, and a conservative SA group size, rather than blindly activating as many segments as possible.

\subsection{Expected Overall Throughput versus Oracle Sampling Size}
\begin{figure}[!t]
\centering
\includegraphics[width=0.4\textwidth]{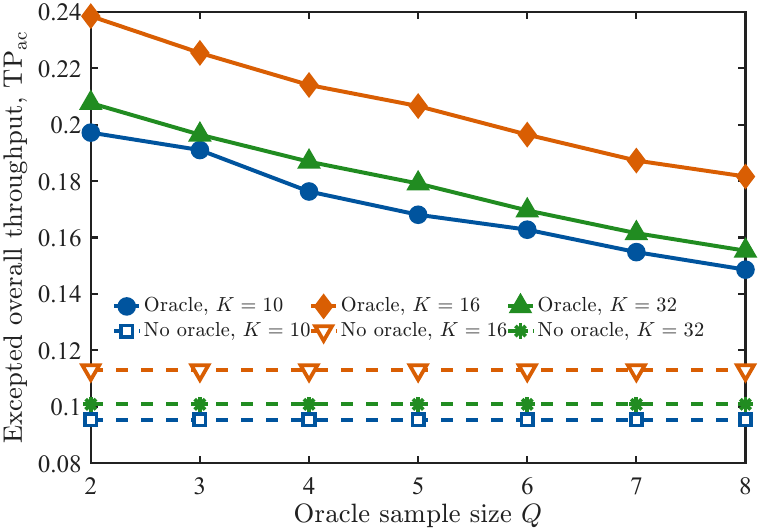}
\caption{Expected overall throughput versus the oracle sampling size $Q_{\mathrm{co}}$ for the oracle-assisted SA protocol and a direct random-SA baseline without channel oracle, with $S=4$, $Q_{\mathrm{ac}}=4$, $\rho_a=10$~dBm, and $\alpha=1$.}
\label{fig_protocol_tp_q}
\vspace{-5pt}
\end{figure}
Fig.~\ref{fig_protocol_tp_q} shows the expected overall throughput versus the oracle sampling size $Q_{\mathrm{co}}$. For all considered values of $K$, the oracle-assisted SA scheme consistently outperforms its no-oracle counterpart, which demonstrates the benefit of incorporating the oracle into the protocol design. For the oracle-assisted scheme, the throughput decreases as $Q_{\mathrm{co}}$ increases, because a larger oracle sampling size mainly prolongs the oracle stage in the considered regime, while bringing only marginal additional access gains. In contrast, the no-oracle baselines are independent of $Q_{\mathrm{co}}$ since they do not involve oracle acquisition. Moreover, as $K$ increases moderately, the throughput improves due to better slot utilization. Once $K$ becomes too large, however, collisions dominate and the throughput decreases.
% Fig.~\ref{fig_protocol_tp_q} shows the expected overall throughput versus the oracle sample size $Q$. For all three values of $K$, the oracle-assisted SA scheme consistently outperforms the corresponding no-oracle baseline, which demonstrates the superiority of the proposed protocol. Moreover, the oracle-assisted throughput decreases with $Q$, since enlarging the oracle sample size mainly increases the oracle duration in the considered setting, while providing only limited additional access gain. By contrast, the no-oracle baselines remain flat because they bypass the oracle stage entirely. Finally, as $K$ increases moderately, the throughput improves because more access opportunities are utilized. However, once $K$ becomes too large, collisions become dominant and the throughput starts to decrease.

\vspace{-5pt}
\subsection{Collision-Aware SA Throughput versus $K$}

\begin{figure}[!t]
\centering
\includegraphics[width=0.4\textwidth]{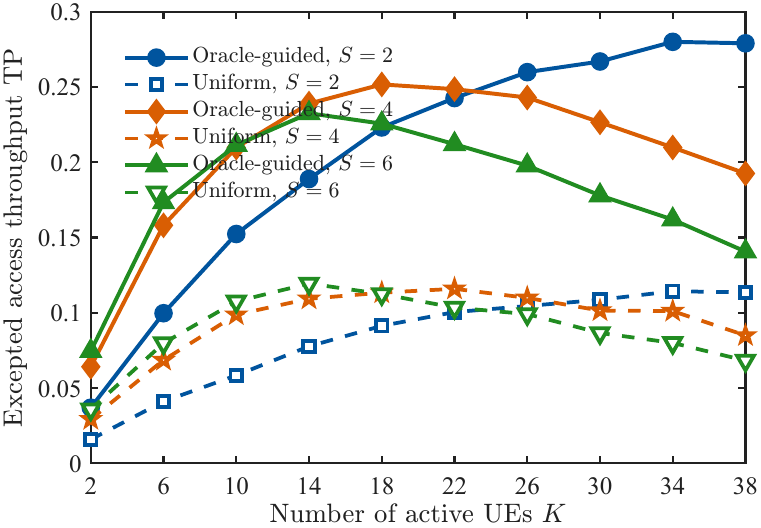}
\caption{Expected access throughput versus the number of active UEs $K$, for oracle-guided and uniform SA-based access scheme with $S\in\{2,4,6\}$.}
\label{fig_sa_throughput_a}
\vspace{-5pt}
\end{figure}

Fig.~\ref{fig_sa_throughput_a} evaluates the collision-aware access throughput of the SA-based access scheme under different group sizes. Compared with uniform slot selection, oracle-guided access improves throughput by steering UEs toward groups with higher predicted SA success probability, rather than treating all groups as equally useful. The resulting throughput is non-monotonic in $K$. It first increases as more access slots are utilized, but eventually decreases when repeated slot conflicts dominate and reduce the number of successful transmissions.
% The curves also exhibit the characteristic ``first-rise-then-fall'' behavior of random access: as $A$ increases, successful slot occupancy first improves and is then overtaken by repeated slot conflicts.

These trends verify the SA design tradeoff characterized in Section~\ref{subsec_det_guarantees}. Decreasing $S$ increases the number of groups and reduces group-level contention as predicted by \eqref{eq_pair_collision_balanced}, but it also reduces the number of segments participating in each analog aggregation. Increasing $S$ has the opposite effect. It may improve aggregation at light load, but reduces contention diversity and makes the single-RF aggregation penalty more pronounced as $K$ grows. The simulation therefore supports the design guideline that SA-based access scheme should employ oracle-guided protocol with a moderate group size rather than aggressively enlarging the aggregation group.
\vspace{-5pt}
\subsection{SA versus $R$-Access Throughput Comparison}

\begin{figure}[!t]
\centering
\includegraphics[width=0.4\textwidth]{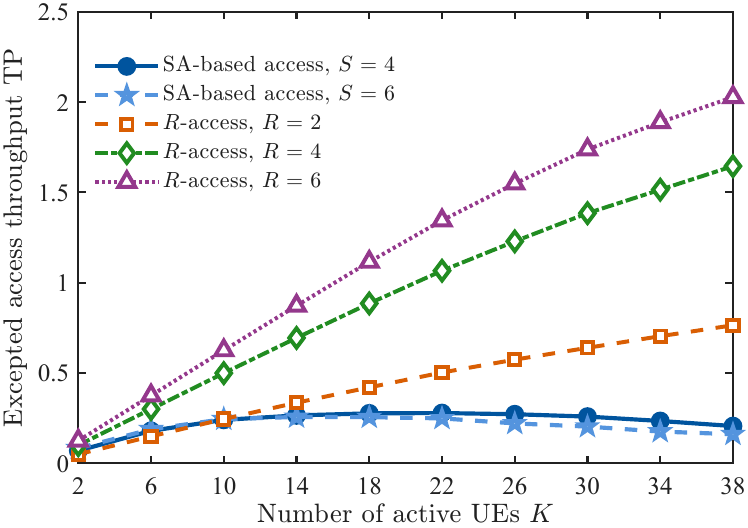}
\caption{Expected access throughput versus the number of active UEs $K$, for the SA scheme with $S\in\{4,6\}$ and the $R$-access scheme with $R\in\{2,4,6\}$.}
\label{fig_sa_vs_raccess_tp}
\vspace{-5pt}
\end{figure}

Fig.~\ref{fig_sa_vs_raccess_tp} reveals a clear architectural separation between the two access mechanisms. Both SA curves are well below the multi-RF $R$-access curves, and the gap widens as $R$ increases, because each SA slot collapses all active segment observations into one analog stream and can resolve at most one successful access event. Accordingly, the SA throughput rises only in the low-load regime and then saturates once the finite group/slot diversity is exhausted, whereas $R$-access preserves multiple parallel RF observations and keeps converting additional offered load into resolved accesses.

The two SA curves also provide a useful intra-architecture insight. A larger group size can offer a slight benefit at very low load, since more segments participate in each analog aggregation. However, the $S=6$ curve degrades faster than the $S=4$ curve once the load becomes moderate, because increasing $S$ simultaneously reduces the number of available contention groups and strengthens the phase-cancellation/noise-threshold penalty already observed in Fig.~\ref{fig_oracle_sa_joint}(c). Therefore, Fig.~\ref{fig_sa_vs_raccess_tp} shows both that SA is throughput-inferior to $R$-access and that SA itself should avoid aggressively large group sizes when multiuser load is the dominant concern. Meanwhile, SA remains attractive when low RF complexity, low power consumption, and simplified implementation are more important than maximizing throughput.

\vspace{-5pt}
\subsection{$R$-Access Coverage Reliability versus Transmit Power}
\begin{figure}[!t]
\centering
\includegraphics[width=0.4\textwidth]{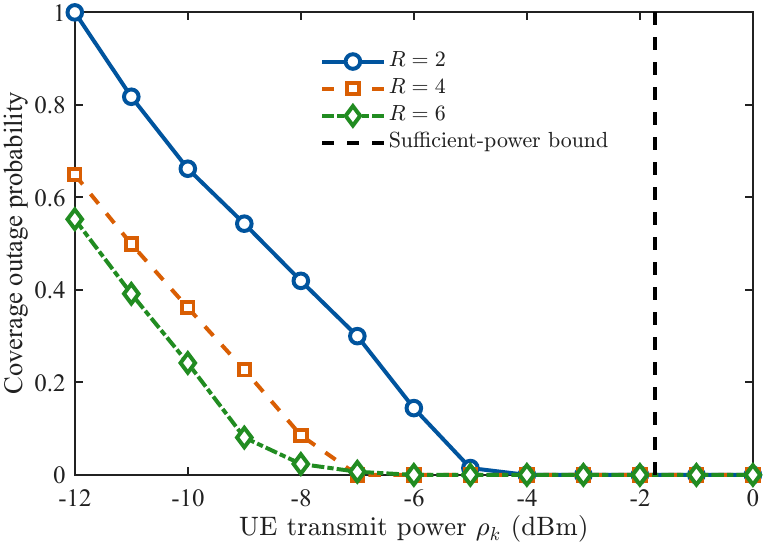}
\caption{Coverage outage probability versus the UE transmit power for the $R$-access scheme with $R\in\{2,4,6\}$. The dashed vertical line denotes the sufficient-power bound implied by Proposition~\ref{prop:partialSM_exist_slot}.}
\label{fig_raccess_guarantee_qac}
\vspace{-8pt}
\end{figure}
Fig.~\ref{fig_raccess_guarantee_qac} provides a more informative view of the deterministic coverage behavior of $R$-access. Instead of plotting a nearly flat best-slot SNR curve against $Q_{\mathrm{ac}}$, the figure directly shows the coverage outage probability as a function of the UE transmit power. A clear monotonic trend is observed for all tested RF chain budgets. As $\rho_k$ increases, the outage probability decreases rapidly and eventually approaches zero. Moreover, increasing $R$ shifts the outage curve to the left, which means that the same target reliability can be achieved with a lower transmit power when more RF chains are available. This behavior is consistent with the strict energy additivity of $R$-access, since a larger $R$ provides more simultaneous digital combining branches and thereby increases the best-slot received energy available to each UE.

The dashed vertical line marks the sufficient-power threshold obtained from \eqref{eq_partialSM_power_suff}. The simulation shows that this bound is conservative yet practically meaningful. Once the transmit power approaches this threshold, the coverage outage becomes negligible for all tested values of $R$, while for larger $R$ the outage decays even earlier. Therefore, the figure not only verifies the theoretical guarantee in Proposition~\ref{prop:partialSM_exist_slot}, but also reveals a concrete design trend that was not visible in the previous $Q_{\mathrm{ac}}$-based plot, namely that multi-RF $R$-access improves coverage robustness primarily by reducing the transmit-power requirement needed to reach near-zero outage.

\vspace{-5pt}
\section{Conclusion}\label{sec_conclusion}
In this paper, an uplink random access protocol was investigated for SWANs under the three canonical operating modes. It was shown that the SWAN configuration-dependent channel variations can be exploited as an access resource for random access design. Based on this insight, an oracle-assisted protocol was developed.
%where sparse OS-based pilot observations were used to infer the configuration-dependent channel responses and guide subsequent uplink access decisions.
For the oracle stage, we formulated a geometry-based channel reconstruction method, and derived an FIM-based characterization that relates the oracle sampling design to the channel reconstruction accuracy. For the access stage, an SA-based scheme was developed for low-complexity single-RF operation, where randomized segment aggregation provides statistical access diversity. Meanwhile, an SM-based $R$-access scheme was designed for scenarios requiring stronger access controllability and multiuser resolution, where a geometry-aware lower bound was established to guarantee feasible access slots under the proposed codebook.
Numerical results confirmed that the proposed framework effectively exploits SWAN reconfigurability for uplink random access. Compared with conventional random access over a fixed propagation environment, the proposed protocol leverages configuration-dependent propagation variations to improve access reliability and throughput.

\vspace{-5pt}
\bibliographystyle{IEEEtran} 
\bibliography{reference}    

\end{document}